\documentclass[final,11pt]{article}

\def\showauthornotes{1}

\usepackage[dvipsnames]{xcolor}
\usepackage[
letterpaper=true,
colorlinks=true,
urlcolor=blue,
linkcolor=blue,
citecolor=OliveGreen,
]{hyperref}
\usepackage{mathrsfs}
\usepackage{fullpage}
\usepackage{amsthm}
\usepackage{amsmath, amssymb}
\usepackage{enumerate}
\usepackage{graphicx}
\usepackage{xspace}
\usepackage{times}
\usepackage{multirow}
\newcommand{\hide}[1]{} 
\usepackage{bm}

\newtheorem{theorem}{Theorem}[section]
\newtheorem{lemma}[theorem]{Lemma}

\newtheorem{definition}[theorem]{Definition}

\newtheorem{claim}[theorem]{Claim}
\newtheorem{remark}[theorem]{Remark}

\newtheorem*{rep@theorem}{\rep@title}
\newcommand{\newreptheorem}[2]{%
\newenvironment{rep#1}[1]{%
 \def\rep@title{#2 \ref{##1}}%
 \begin{rep@theorem}}%
 {\end{rep@theorem}}}
\makeatother

\newreptheorem{lemma}{Lemma}
\newreptheorem{theorem}{Theorem}

\DeclareMathOperator*{\MFN}{\mathrm{MFN}}
\newcommand{\MFNLP}[0]{\textsc{MFN-LP}\xspace}
\DeclareMathOperator*{\LPdemand}{\textnormal{LP}_{\textnormal{demand}}}

\renewcommand{\vec}[1]{{\bm{#1}}}
\newcommand{\vx}{{\bm{x}}}
\newcommand{\vy}{{\bm{y}}}
\newcommand{\vg}{{\bm{g}}}

\newcommand{\cA}{\mathcal A}

\newcommand{\cD}{\mathcal D}

\newcommand{\cF}{\mathcal F}

\newcommand{\cI}{\mathcal I}

\newcommand{\cP}{\mathcal P}

\newcommand{\cfl}{\textsc{Cfl}\xspace}
\newcommand{\ufl}{\textsc{Ufl}\xspace}

        {\hspace*{\fill}$\Box$\par}

\DeclareMathOperator{\poly}{poly}

\newcommand{\ie}{i.e.,\xspace}

\ifnum\showauthornotes=1
\newcommand{\Authornote}[2]{{\sffamily\small\color{red}{[#1: #2]}}}
\newcommand{\Authornotecolored}[3]{{\sffamily\small\color{#1}{[#2: #3]}}}
\newcommand{\Authorcomment}[2]{{\sffamily\small\color{gray}{[#1: #2]}}}
\newcommand{\Authorstartcomment}[1]{\sffamily\small\color{gray}[#1: }

\newcommand{\Authorfnote}[2]{\footnote{\color{red}{#1: #2}}}
\newcommand{\Authorfixme}[1]{\Authornote{#1}{\textbf{??}}}
\newcommand{\Authormarginmark}[1]{\marginpar{\textcolor{red}{\fbox{\Large #1:!}}}}
\else
\newcommand{\Authornote}[2]{}
\newcommand{\Authornotecolored}[3]{}
\newcommand{\Authorcomment}[2]{}
\newcommand{\Authorstartcomment}[1]{}

\newcommand{\Authorfnote}[2]{}
\newcommand{\Authorfixme}[1]{}
\newcommand{\Authormarginmark}[1]{}
\fi

\usepackage{tikz}
\usepackage{pgfplots}
\usetikzlibrary{plotmarks}
\usetikzlibrary{shadows}
\usetikzlibrary{decorations.pathreplacing}
\usetikzlibrary{positioning}
\usetikzlibrary{scopes}
\usetikzlibrary{arrows}
\usetikzlibrary{calc}

\usetikzlibrary{backgrounds}

\tikzstyle{client}=[circle, draw=black,fill=white, inner sep=0pt, minimum size=20pt]
\tikzstyle{facility}=[draw=black,fill=white, inner sep=0pt, minimum size=18pt]
\tikzstyle{sclient}=[circle, draw=black,fill=white, inner sep=0pt, minimum size=6pt]
\tikzstyle{sfacility}=[draw=black,fill=white, inner sep=0pt, minimum size=5pt]

\tikzset{>= stealth'}

\title{\bfseries LP-Based Algorithms for Capacitated Facility Location}

\author{%
Hyung-Chan An\thanks{School of Computer and Communication Sciences, EPFL.
Email:
\href{mailto:hyung-chan.an@epfl.ch}{hyung-chan.an@epfl.ch}}, Mohit Singh\thanks{Microsoft Research,
Redmond. Email: \href{mailto:mohits@microsoft.com}{mohits@microsoft.com}.} , Ola
Svensson\thanks{School of Computer and Communication Sciences, EPFL.
Email:
\href{mailto:ola.svensson@epfl.ch}{ola.svensson@epfl.ch}. Supported by ERC Starting Grant 335288-OptApprox.}
}

\date{}
\begin{document}
\setcounter{page}{0}
\maketitle

\begin{abstract}
Linear programming has played a key role in the study of algorithms for
combinatorial optimization problems. In the field of approximation algorithms, this is well
illustrated by the uncapacitated facility location problem. A variety of algorithmic methodologies,
such as LP-rounding and primal-dual method, have been applied to and evolved from algorithms for this
problem. Unfortunately, this collection of powerful algorithmic techniques had not yet been applicable to the
more general
capacitated facility location problem. In fact, all of the known algorithms with good performance
guarantees were based on a single technique, local search, and no linear programming relaxation was known to efficiently approximate the problem.

In this paper, we present a linear programming relaxation with constant integrality gap for
capacitated facility location.
We demonstrate that the fundamental theories of multi-commodity flows and matchings provide key insights that lead to the strong relaxation. Our
algorithmic proof of integrality gap is obtained by finally accessing the rich toolbox of LP-based
methodologies: we present a constant factor approximation algorithm based on LP-rounding.


\end{abstract}

\thispagestyle{empty}
\newpage

\section{Introduction}
We consider the metric capacitated facility location (\cfl) problem which together with the metric
uncapacitated facility location (\ufl) problem is the most classical and widely
studied variant of facility location.  In \cfl, 
we are given a single metric on the set of \emph{facilities} and \emph{clients}, and every facility
has an associated \emph{opening cost} and \emph{capacity}. The problem asks us to choose a subset of
facilities to open and assign every client to one of these open facilities, while ensuring that no
facility is assigned more clients than its capacity. Our aim is then to find a set of open
facilities and an assignment that minimize the cost, where the cost is defined as the sum of opening
costs of each open facility and the distance between each client and the facility it is assigned to.
\ufl is the special case of \cfl obtained by dropping the capacity constraints, or equivalently
setting each capacity to~$\infty$.

In spite of the similarities in the problem definitions of \ufl and \cfl,
current techniques give a considerably better understanding of the uncapacitated version.
One prominent reason for this discrepancy is that a standard linear programming (LP)
relaxation gives close-to-tight bounds for \ufl, whereas no good relaxation
was known in the presence of capacities.
 For \ufl, on the one hand,  the standard LP formulation has been used in combination with most LP-based techniques, such
as filtering~\cite{LV92B}, randomized rounding~\cite{CS04,STA97}, primal-dual framework~\cite{JV01},
and dual fitting~\cite{JMM03,JMS02}, to obtain a fine-grained understanding of the problem resulting
in a nearly tight approximation ratio~\cite{Li11}.

For \cfl, on the other hand, it has remained a major open problem to find a relaxation based
algorithms with
\emph{any} constant performance guarantee, also highlighted as Open Problem 5 in the list of ten open
problems selected by the recent textbook on approximation algorithms of Williamson and Shmoys~\cite{WS10}. This
question is especially intriguing as there exist constant factor approximation algorithms for \cfl
based on the local search paradigm (see e.g.~\cite{ALBGGGJ13,BGG12,CW05,KPR00,KH63,PTW01} and Section~\ref{sec:relatedwork} for further
discussion of this approach).
Compared to such methods, there are several advantages of  algorithms based on relaxations. First,
as alluded to above, there
is a large toolbox of LP-based techniques that one can tap into once a strong relaxation is
known for a problem. Second, LP-based algorithms give a stronger \emph{per instance} guarantee: that is, the
rounded solution is compared to the found LP solution and the gap is often smaller than the worst
case. This is in contrast to local search heuristics that only guarantees that the cost is no worse than
the proven \emph{a priori} performance guarantee assures. Finally, LP-based techniques are often flexible
and allow for generalizations to related problems. This has indeed been the case for the
uncapacitated versions where algorithms for \ufl are used in the design of approximation algorithms
for other related problems, see for example~\cite{ByrkaSS10,JV01,KrishnaswamyKNSS11,LS13}.

In this pursuit of LP-based approximation algorithms for capacitated facility location problems, the central question lies in devising a strong LP relaxation that is  \emph{algorithmically amenable}.
In fact, any combinatorial problem has a relaxation with constant integrality gap: the exact formulation,
which is the convex hull of integral solutions, has indeed an integrality gap of 1.  However, such 
formulations for NP-hard optimization problems usually have insufficient structure  to give enough insights for
designing efficient approximation algorithms. The challenge, therefore, is to instead devise an LP-relaxation that is
sufficiently strong while featuring enough structure so as to guide the development of efficient
approximation algorithms using LP-based techniques, such as LP-rounding or primal-dual.
However, formulating one for the capacitated facility location problem 
has turned out to be non-trivial.
Aardal et al.~\cite{APW95} made a
comprehensive study of valid inequalities for capacitated facility location problem and proposed further generalizations;
the  strength of the obtained formulations was left as an
open problem.
Many of these formulations were, however, recently proven to be insufficient for obtaining a
constant integrality gap by Kolliopoulos and Moysoglou~\cite{KM13b}. In the same paper it is also
shown that applying the Sherali-Adams hierarchy to the standard LP formulation will not close the
integrality gap.

\subsection{Our Contributions}
\label{sec:ourcontrib}
Our main contribution is a strong linear programming relaxation which has a constant integrality gap for the capacitated facility location problem. We prove its constant integrality gap by presenting a polynomial time approximation algorithm which rounds the LP solution.

\begin{theorem}\label{thm:main}
There is a linear programming relaxation (\MFNLP given in Figure~\ref{fig:MFN-LP}) for the capacitated facility location problem that has a constant integrality gap. Moreover, there exists a polynomial-time algorithm that finds a solution to the capacitated facility location problem whose cost is no more than a constant factor times the LP optimum.
\end{theorem}

\noindent
This result resolves Open Problem 5 in the list of ten open problems selected by the textbook of Williamson and Shmoys~\cite{WS10}.

Our relaxation is formulated based on multi-commodity flows. We will discuss in this section why the multi-commodity flow is a natural tool of choice in designing strong LP relaxations for our problem, and also how it plays a key role, together with the matching theory, in achieving a constant factor LP-rounding algorithm.

One natural question that arises is characterizing the exact integrality gap of our relaxation. While we prioritized ease of reading over a better ratio in the choice of parameters for this presentation of our algorithm, it appears that the current analysis is not likely to give any approximation ratio better than $5$, the best ratio given by the local search algorithms~\cite{BGG12}. On the other hand, the best lower bound known on the integrality gap of our relaxation is $2$, and
 the question remains open whether we can obtain an approximation algorithm with a ratio smaller than $5$ based on our relaxation.
\\

\noindent
\textbf{Open Question.}
Determine the integrality gap of the LP relaxation \MFNLP.

\paragraph{High-level description of \MFNLP.}
The minimum knapsack problem is a special case of capacitated facility location: given a target value and a set of items with different values and costs, the problem is to find a minimum-cost subset of items whose total value is no less than the given target. Carr et al.~\cite{CFLP2000} showed that flow-cover inequalities~\cite{W, CFLP2000} yield an LP with a constant integrality gap for this problem; in fact, another aspect of our relaxation shares a similar spirit as these inequalities. The flow-cover inequalities for the minimum knapsack problem say that, when any subset of items is given for ``free'' to be part of the solution, the LP solution should be feasible to the residual problem. In this residual problem, the target value is decreased by the total target value of the free items; hence, constraints of the residual problem can be strengthened by updating the values of all items to be at most the new target value.

In order to have a similar notion of residual sub-problems in the facility location problem, it is tempting to formulate a sub-problem for each subset of facilities which are open for free. Indeed the knapsack problem suggests exactly this sub-problem, since in the reduction from the knapsack problem to the facility location problem, items correspond to facilities. However, we take a different approach. Observe that there are two types of decisions to be made in the facility location problem: which facilities to open, and how to assign the clients to these open facilities; we focus on the latter. We contemplate an assignment of a subset of clients to some facilities, and insist that this assignment should be a part of the solution. We formulate the residual problem on unassigned clients, update the capacity of each facility and reduce it by the number of clients assigned for free to this facility. We now require that any feasible solution to the problem must contain a feasible solution to the residual problem. We call the assignment of clients for free a \emph{partial assignment}, as they assign only a subset of clients.

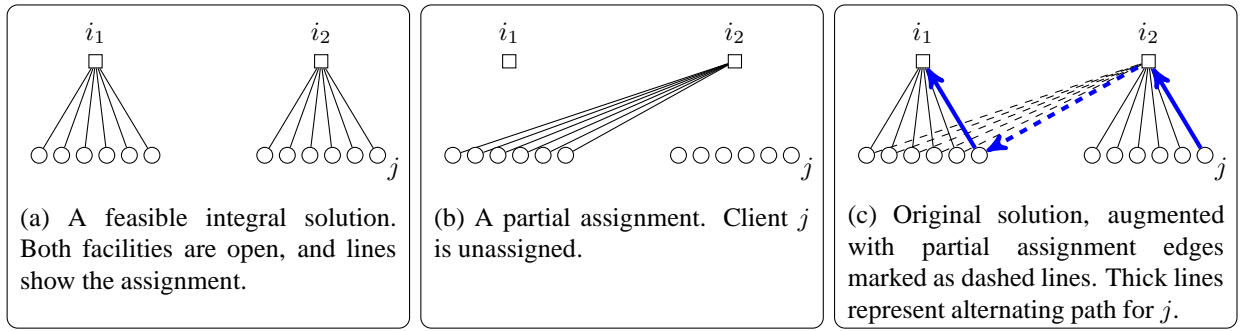
\begin{figure}[t]
\centering
\begin{tikzpicture}
\begin{scope}[scale=0.5]
\draw[rounded corners] (-5.35,-4.6) rectangle (5.35,4);
\node at (0,-2.9) {\begin{minipage}{5cm}\small (a) A feasible integral solution. Both facilities are open, and lines show the assignment.\\ \ \end{minipage}};

\draw node[sfacility] (i1) at (-3,2.5) {};
\node[above = 0.1cm] at (i1) {\small $i_1$};
\draw node[sfacility] (i2) at (3,2.5) {};
\node[above = 0.1cm] at (i2) {\small $i_2$};

\foreach \vali in {1, ..., 6} {
  \draw node[sclient] (a\vali) at (-5.1+ \vali*0.6, 0) {};
  \draw (a\vali) edge (i1);
}

\foreach \vali in {1, ..., 6} {
  \draw node[sclient] (b\vali) at (0.9+ \vali*0.6, 0) {};
  \draw (b\vali) edge (i2);
}
\node[below right = -0.1cm and 0cm] at (b6) {\small $j$};
\end{scope}

\begin{scope}[scale=0.5,xshift=11cm]
\draw[rounded corners] (-5.35,-4.6) rectangle (5.35,4);
\node at (0,-2.9) {\begin{minipage}{5cm}\small (b) A partial assignment.  Client $j$ is unassigned.\\ \ \\ \ \end{minipage}};

\draw node[sfacility] (i1) at (-3,2.5) {};
\node[above = 0.1cm] at (i1) {\small $i_1$};
\draw node[sfacility] (i2) at (3,2.5) {};
\node[above = 0.1cm] at (i2) {\small $i_2$};

\foreach \vali in {1, ..., 6} {
  \draw node[sclient] (a\vali) at (-5.1+ \vali*0.6, 0) {};
  \draw (a\vali) edge (i2);
}

\foreach \vali in {1, ..., 6} {
  \draw node[sclient] (b\vali) at (0.9+ \vali*0.6, 0) {};
}
\node[below right = -0.1cm and 0cm] at (b6) {\small $j$};
\end{scope}

\begin{scope}[scale=0.5,xshift=22cm]
\draw[rounded corners] (-5.35,-4.6) rectangle (5.35,4);
\node at (0,-2.9) {\begin{minipage}{5cm}\small (c) Original solution, augmented with partial assignment edges marked as dashed lines. Thick lines represent alternating path for $j$. \end{minipage}};
\draw node[sfacility] (i1) at (-3,2.5) {};
\node[above = 0.1cm] at (i1) {\small $i_1$};
\draw node[sfacility] (i2) at (3,2.5) {};
\node[above = 0.1cm] at (i2) {\small $i_2$};

\foreach \vali in {1, ..., 5} {
  \draw node[sclient] (a\vali) at (-5.1+ \vali*0.6, 0) {};
  \draw (a\vali) edge (i1);
  \draw[dashed] (a\vali) edge (i2);
}
\foreach \vali in {6} {
  \draw node[sclient] (a\vali) at (-5.1+ \vali*0.6, 0) {};
  \draw[ultra thick,blue] (a\vali)[->] edge (i1);
  \draw[dashed, ultra thick, blue] (a\vali) edge[<-] (i2);
}

\foreach \vali in {1, ..., 5} {
  \draw node[sclient] (b\vali) at (0.9+ \vali*0.6, 0) {};
  \draw (b\vali) edge (i2);
}
\foreach \vali in {6} {
  \draw node[sclient] (b\vali) at (0.9+ \vali*0.6, 0) {};
  \draw[ultra thick, blue] (b\vali) edge[->] (i2);
}
\node[below right = -0.1cm and 0cm] at (b6) {\small $j$};
\end{scope}

\end{tikzpicture} 
\caption{Example of partial assignment. Squares represent facilities of capacity 6; circles clients.}
\label{fig:example}
\end{figure}

While the residual instance would be again an instance of the capacitated facility location, with fewer clients and facilities with reduced capacities, it is not clear whether restricting a feasible solution of the original problem forms a feasible solution to the residual problem. In fact, it does not. The partial assignment reduces capacities at facilities which the feasible solution might have used for clients remaining in the residual instance. To be concrete, consider a feasible integral solution depicted in Figure~\ref{fig:example}a and a partial assignment in Figure~\ref{fig:example}b. Note that client $j$ was not assigned by the partial assignment, but in the residual instance, it cannot be assigned to facility $i_2$ as the original solution indicates. The partial assignment has already assigned enough clients to reduce the capacity of facility $i_2$ to zero in the residual instance. But observe that the fact that client $j$ could not claim its original place means that some other client has taken its place; therefore, that client must have left behind its
space somewhere else (at facility $i_1$ in this example). Thus we would want to assign client $j$ to facility $i_1$ in the example. But how can we enable
such an assignment in general?
Our relaxation allows additional edges to be used for assignments in the residual instance. In particular, we make edges corresponding to the partial assignment available to be used to ``undo'' the partial assignment;
observe that what we are now looking for is not a direct assignment of clients to facilities but \emph{alternating paths} starting at each client in the residual instance to a facility with spare reduced capacity.
We model this problem as a \emph{multi-commodity flow} problem where every unassigned client demands a unit flow to be routed to a facility with residual capacity.
In fact, it is crucial for obtaining a strong LP to use multi-commodity flows to model these assignments, as we will see in Section~\ref{sec:lp}.

For the interested reader, we further relate our relaxation to previous works by demonstrating in Section~\ref{sec:knapsack} how it automatically embraces two interesting special cases, including the flow-cover inequalities for the knapsack problem. In addition, we show how our relaxation deals with a specific instance, for which the standard LP has an unbounded integrality gap.

\paragraph{LP-rounding algorithm.}
In Section~\ref{sec:algorithm}, we give an algorithmic proof of constant integrality gap by presenting a polynomial-time LP-rounding algorithm. An interesting feature of this algorithm is that it does not solve the LP to optimality. Instead, we will give a rounding procedure that either rounds a given fractional solution within a constant factor, or identifies a violated inequality. This approach has been previously used, see for example, Carr et al.~\cite{CFLP2000} and Levi et al.~\cite{LeviLS07}.

As is the case for flow-cover inequalities, we do not know whether our relaxation can be separated in polynomial time. However, our rounding algorithm establishes that it suffices to separate it over a given partial assignment in order to obtain a constant approximation algorithm: in a sense, such limited separation is already enough to extract the power of our strong relaxation within a constant factor.
That said, it remains an interesting open question whether our relaxation can be separated in polynomial time. Another interesting open question would be whether there exists a different LP relaxation that can be solved in polynomial time and used to design a constant approximation algorithm.

Given a fractional solution consisting of \emph{opening variables} and \emph{assignment variables}, the first step our rounding algorithm takes is very natural: we decide to open all the facilities whose opening variables are large, say, at least $\frac12$. The cost of opening these facilities is no more than twice the cost paid by the fractional solution. Now, we find an assignment of \emph{maximum} number of clients to these facilities while maintaining that the assignment cost does not exceed twice the cost of fractional solution. If we manage to assign all clients to the integrally opened facilities, we are done since both the connection cost and facility opening cost can be bounded within a constant factor of the linear programming solution. Else, we obtain a partial assignment of clients to the opened facilities. We use this partial assignment
to formulate the multi-commodity flow problem described earlier. Recall, in the multi-commodity flow problem, each unassigned client has a flow commodity which it needs to sink at the facilities using alternating paths. Assume for simplicity, that in the partial assignment all facilities that we opened in the first step are saturated.
Now, in the multi-commodity flow problem, a client can only sink flow at facilities with small fractional value because the facilities with large fractional value have zero capacity since they are saturated by the partial assignment. Thus, the flow solution naturally gives us a fractional assignment of remaining unassigned clients to facilities which are open to a small fractional value. In the last step of the algorithm, we round this fractional solution obtained via the flow problem. But why is this problem any easier than the one we started with? Since each facility opening variable is at most $\frac12$, the fractional solution
can use
at most half the capacity of any facility in the residual instance. Thus the capacity constraints are not stringent and we can appeal to known soft-capacitated approximation algorithms which approximate cost while violating capacity to a small factor (two suffices for us). Indeed, such algorithms can be obtained by rounding the standard linear program and we use the result of Abrams et al.~\cite{AMMP2002}. This also implies that an immediate improvement to the approximation ratio of our algorithm would be possible by providing an improved algorithm for the soft-capacitated problem.

In summary, we
have used techniques from the theory of multi-commodity flows and matchings to formulate the first
linear programming relaxation for the capacitated facility location problem that efficiently approximates the optimum value
within a constant. Our proposed LP-rounding algorithm exploits the properties of the multi-commodity flows obtained by solving the linear program and we give a constant factor approximation algorithm for the problem.
Our results further open up the possibility to approach the capacitated facility location problem and other related problems using the large
family of known powerful LP-based techniques.


\subsection{Further Related Work}
\label{sec:relatedwork}

\paragraph{Uncapacitated facility location.}
Since the first constant factor approximation algorithm for \ufl was given by Shmoys, Tardos and Aardal~\cite{STA97}, several techniques have been developed around this problem. Currently, the best approximation guarantee of $1.488$ is due to Li~\cite{Li11}; see also ~\cite{BA10,JMS02}.
On the hardness side, Guha and Khuller~\cite{GK99} shows that it is hard to approximate \ufl within a factor of 1.463.

\paragraph{Local search heuristics for capacitated facility location.} All
previously known constant factor
approximation algorithms for \cfl are based on the local
search paradigm. The first constant factor approximation algorithm was obtained in the special case
of uniform capacities (all capacities being equal) by Korupolu et al.~\cite{KPR00} who analyzed a
simple local search heuristic proposed by Kuehn and Hamburger~\cite{KH63}. Their analysis was then
improved by Chudak and Williamson~\cite{CW05} and the current best $3$-approximation algorithm for
this special case is a local search by Aggarwal et al.~\cite{ALBGGGJ13}. For the general problem
(\cfl), P\'{a}l et al~\cite{PTW01} gave the first constant factor approximation algorithm. Since
then more and more sophisticated local search heuristics have been proposed, the current best being a
recent local search by Bansal et al.~\cite{BGG12} which yields a $5$-approximation algorithm.

\paragraph{Relaxed notions of capacity constraints.} Several special cases or relaxations of
the capacitated facility location problem have been studied. One popular relaxation is the soft-capacitated problems where the capacity constraints are relaxed in various ways. The standard linear program still gives a good bound for many of these relaxed problems. Shmoys et al.~\cite{STA97} gives the first constant factor approximation
algorithm where a facility is allowed to be open multiple times, later improved by Jain and
Vazirani~\cite{JV01}. Mahdian et al.~\cite{MYZ03} gives the current best approximation ratio of $2$,
which is tight with respect to the standard LP. Abrams et al.~\cite{AMMP2002} studies a variant
where a facility can be open at most once, but the capacity can be violated by a constant factor.
We also mention that in our approximation algorithm, we use this variant of relaxed capacities as a subproblem.
Finally, another special case for which the standard LP has been amenable to is the case of uniform opening
costs, \ie when all facilities have the
same opening cost. For that case, Levi et al.~\cite{LSS12} gives a 
$5$-approximation algorithm.

We also mention that LP-based approximation algorithms which do not solve the linear program to optimality have been used in the works of Carr et. al~\cite{CFLP2000} and Levi et. al.~\cite{LeviLS07}. In a similar spirit, many primal-dual algorithms do not solve linear programs to optimality (see e.g.~\cite{AKR, GW}), while finding approximate solutions whose guarantee is given by comparison to a feasible dual solution.

Finally, we note that Chakrabarti, Chuzhoy and Khanna~\cite{CCK} used a collection of flow problems to obtain improved approximation algorithms for the max-min allocation problem.


\section{Multi-commodity Flow Relaxation}\label{sec:lp}
We present our new relaxation for the capacitated facility location problem in this section. Let us first define some notation to be used in the rest of this paper. Let $\cF$ be the set of facilities and $\cD$ be the set of clients. Each facility $i\in\cF$ has opening cost $o_i$, and cannot be assigned more number of clients than its capacity $U_i$. We are also given a metric cost $c$ on $\cF\cup\cD$ as a part of the input: $c_{ij}$ denotes the distance between $i\in\cF$ and $j\in\cD$.

The variables of our relaxation is the pair $(\vec{x}, \vec{y})$ where we refer to $\vec{x}\in
[0,1]^{\cF \times \cD}$ as the \emph{assignment variables} and to $\vec{y}\in [0,1]^{\cF}$ as the \emph{opening variables}. These variables naturally encode the decisions to which facility a client is connected
and which facilities that are opened. Indeed, the intended integral solution is that $x_{ij} =1$ if
client $j$ is connected to facility $i$ and $x_{ij} = 0$ otherwise; $y_i=1$ if facility $i$ is
opened and $y_i = 0$ otherwise. The idea of our relaxation is that every partial assignment of
clients to facilities should be extendable to a complete assignment while only using the assignments
of $\vec{x}$ and openings of $\vec{y}$. To this end let us first describe the partial assignments
that we shall consider. We then define the constraints of our linear program which will be
feasibility constraints of multi-commodity flows.

A partial fractional assignment $\vec{g}=\{g_{ij}\}_{i\in \cF, j \in \cD}$ of clients to facilities is  \emph{valid} if
$$
\forall j\in
\cD: \sum_{i\in \cF} g_{ij} \leq 1   \mbox{,} \qquad  \forall i \in \cF: \sum_{j\in \cD}
g_{ij} \leq U_i \qquad \mbox{and} \qquad  \forall i \in \cF, j\in \cD: g_{ij}\geq 0.
$$
The first condition says that each client should be fractionally assigned at most once and the
second condition says that no facility should receive more clients than its capacity. We emphasize
that we allow clients to be fractionally assigned, \ie $\vec{g}$ is not assumed to  be integral. As we
shall see later (see Lemma~\ref{lem:bing}), this does not change the strength of our relaxation but it will
be convenient in the analysis of our rounding algorithm in Section~\ref{sec:algorithm}. We also remark that
the above inequalities are exactly the $b$-matching polytope of the complete bipartite graph
consisting of the clients on the one side and the facilities on the other side; each client can be
matched to at most one facility and each facility $i$ can be matched to at most $U_i$ clients.

The constraints of our relaxation will be that, no matter how we partially assign the clients according to a valid
$\vec{g}$, the solution $(\vec{x},\vec{y})$ should support a multi-commodity flow where each client $j$ becomes the source of its own commodity $j$, and the demand of this commodity is equal to the amount by which $j$ is ``not assigned'' by $\vec g$, $1-\sum_{i\in \cF} g_{ij}$. The flow network, whose arc capacities are given as a function of $\vec{g}$ and the solution $(\vec{x},\vec{y})$, is defined as follows (see also Figure~\ref{fig:MFN}):
\begin{definition}[Multi-commodity flow network]
\label{def:mfn}
For a valid partial assignment $\vec{g}$,  assignment variables $\vec{x} = \{x_{ij}\}_{i\in \cF, j\in \cD}$, and opening variables $ \vec{y}=
\{y_i\}_{i\in \cF}$, let $\MFN(\vec{g},\vec{x},\vec{y})$ be a
multi-commodity flow network with $|\cD|$ commodities, defined as follows. Note that some arcs may have zero capacities.
\begin{enumerate}[(a)]
\item Each client $j\in \cD$ is associated with commodity $j$ of demand $d_j := 1 - \sum_{i\in \cF} g_{ij}$, and its source-sink pair is $(j^s, j^t)$.  \item Each facility $i\in \cF$ has two nodes $i$ and $i'$ with an arc $( i,i')$ of capacity $y_i\cdot( U_i - \sum_{j\in \cD} g_{ij})$.  \item For each client $j$ and facility $i$, there is an arc $( j^s, i)$ of capacity $x_{ij}$, an arc
  $( i, j^s) $ of capacity $g_{ij}$, and an arc $( i', j^t) $ of capacity $y_i d_j$.
\end{enumerate}
\end{definition}

\begin{figure}[tb]
\begin{center}
\begin{tikzpicture}

\node[client] (j1) at (-1.5, 0)  {$j^s_1$};
\node[client] (j2) at (-1.5, -1.5)  {$j^s_2$};
\node at (-1.5,-3) {$\vdots$};
\node[client] (jn) at (-1.5, -4.5)  {$j^s_n$};

\node (d) at (-1.7, 1.5) {\small demand $d_{j_1} = 1 - \sum_{i\in \cF} g_{ij_1}$};

\draw (d) edge[dotted] (j1);

\node[facility] (i1) at (1.5,-0.5) {$i_1$};

\node[facility] (i2) at (1.5,-1.75) {$i_2$};

\node at (1.5,-2.9) {$\vdots$};
\node[facility] (im) at (1.5,-4) {$i_m$};

\node[facility] (i'1) at (6.5,-0.5) {$i'_1$};

\node[facility] (i'2) at (6.5,-1.75) {$i'_2$};

\node at (6.5,-2.9) {$\vdots$};
\node[facility] (i'm) at (6.5,-4) {$i'_m$};

\node[client] (t1) at (9.7,0 )  {$j^t_1$};
\node[client] (t2) at (9.7,-1.5 )  {$j^t_2$};
\node at (9.7,-3 ) {$\vdots$};
\node[client] (tn) at (9.7,-4.5 )  {$j^t_n$};

\draw (j1) edge[->,ultra thick, bend left] node[above] {\small capacity $x_{i_1j_1}$}   (i1);
\draw[gray] (j1) edge[<->] (i2);
\draw[gray] (j1) edge[<->] (im);

\draw[gray] (j2) edge[<->] (i1);
\draw[gray] (j2) edge[<->] (i2);
\draw[gray] (j2) edge[<->] (im);

\draw[gray] (jn) edge[<->] (i1);
\draw[gray] (jn) edge[<->] (i2);
\draw[gray] (jn) edge[->] (im);

\draw (jn) edge[<-, ultra thick, bend right] node[below] {\small capacity $g_{i_mj_n}$}  (im);

\draw (i1) edge[ultra thick,->] node[above] {\small\begin{tabular}{c}capacity\\$y_{i_1} \cdot (U_{i_1} -  \sum_{j\in \cD} g_{i_1 j})$\end{tabular}} (i'1);
\draw[gray] (i2) edge[->] (i'2);
\draw[gray] (im) edge[->] (i'm);

\draw (i'1) edge[ultra thick,->] node[above=0.2] {\small capacity $ y_{i_1} d_{j_1}$}(t1);
\draw[gray] (i'1) edge[->] (t2);
\draw[gray] (i'1) edge[->] (tn);

\draw[gray] (i'2) edge[->] (t1);
\draw[gray] (i'2) edge[->] (t2);
\draw[gray] (i'2) edge[->] (tn);

\draw[gray] (i'm) edge[->] (t1);
\draw[gray] (i'm) edge[->] (t2);
\draw[gray] (i'm) edge[->] (tn);

  \begin{pgfonlayer}{background}
    \filldraw [line width=4mm,join=round,black!10]
      (-1.8,0.7)  rectangle (1.8,-5.2);
  \end{pgfonlayer}

\end{tikzpicture}
\end{center}

\caption{A depiction of the multi-commodity flow network $\MFN(\vec{g},\vec{x},\vec{y})$.}
\label{fig:MFN}
\end{figure}

\begin{remark}
Intuitively, the bipartite subgraph induced by $\{j^s\}_{j\in\cD}\cup\{i\}_{i\in\cF}$, marked with a shaded box in Figure~\ref{fig:MFN}, is the interesting part of the flow network. $\{i'\}_{i\in\cF}$ and $\{j^t\}_{j\in\cD}$ are added to this bipartite graph purely in order to state that $i$ is a sink with ``double'' capacities: a commodity-oblivious capacity $y_i\cdot( U_i - \sum_{j\in \cD} g_{ij})$ and a commodity-specific capacity $y_i d_j$ for each client $j\in \cD$.

\end{remark}

Let us give some intuition on the definition of $\MFN(\vec{g},\vec{x},\vec{y})$. As already noted,
the demand $d_j = 1 - \sum_{i\in \cF} g_{ij}$ of a client $j$ equals the amount by which $j$ is not
assigned by the partial assignment $\vec{g}$. This demand should only be assigned to opened
facilities. Therefore, facility $i$ can accept at most $y_i d_j$ of $j$'s demand which is either
$d_j$ or $0$ in an integral solution.
Observe that such a constraint, for each client and facility, cannot be imposed by a single-commodity flow problem.
Multi-commodity flow problems, on the other hand, allows us to express this constraint as a commodity-specific capacity of $y_id_j$, as denoted by arc $(i',j^t)$ in Figure~\ref{fig:MFN}.

Now consider the \emph{commodity-oblivious} capacities of the facilities. The \emph{total} demand an opened facility $i$ can accept
is its capacity minus the amount of clients assigned to it in the partial assignment $\vec{g}$; and
a closed facility can accept no demand. Therefore, the total demand a facility $i$ can accept is at
most $y_i (U_i - \sum_{j\in \cD} g_{ij})$. The arc capacity $x_{ij}$ of an arc $(j,i)$ says that
client $j$ should be connected to facility $i$ only if $x_{ij} =1$. The reason for having arcs of
the form $(i,j)$ of capacity $g_{ij}$ is discussed in Section~\ref{sec:ourcontrib}: these allow the alternating paths for routing the remaining demand and are
necessary for the formulation to be a relaxation.

We are now ready to formally state our relaxation \MFNLP of the  capacitated facility
location problem in Figure~\ref{fig:MFN-LP}. Note that the only variables of our relaxation are the assignment variables $\vec x$ and the opening variables $\vec y$. While it is natural to formulate each of the multi-commodity flow problem using auxiliary variables denoting the flow, our algorithm will utilize the equivalent formulation obtained via projecting out the flow variables. This projected formulation is a relaxation where the only variables are assignment variables $\vec x$ and the opening variables $\vec y$.
\begin{figure}[h!]
\begin{equation*}
\boxed{%
        \begin{minipage}{10cm}%
          \begin{align*}
            \textrm{minimize}\qquad & c(\vx,\vy):= \sum_{i\in \cF} o_i\cdot y_i  + \sum_{i\in \cF, j\in \cD} c_{ij} \cdot x_{ij},\\
            \textrm{subject to}\qquad & \MFN(\vec{g}, \vec{x}, \vec{y}) \mbox{ is feasible}\quad \forall \vec g\textrm{ valid};\\[2mm]
            &\vec{x}\in [0,1]^{\cF\times\cD},\vec{y}\in[0,1]^{\cF}.
          \end{align*}
          \vspace{-0.4cm}
        \end{minipage}%
}
\end{equation*}
\caption{Our relaxation of \cfl.}
\label{fig:MFN-LP}
\end{figure}

In Lemma~\ref{lem:bing} we show that the constraints of \MFNLP can equivalently be formulated over
the subset of valid partial assignments that are integral. \MFNLP can therefore be seen as the
intersection of the feasible regions
of finitely many multi-commodity flow linear programs and is therefore itself a linear program.
At first sight, however, it may not be clear that \MFNLP is a relaxation, or how we can separate it.
We will answer these questions in the rest of this section.

\subsection{Integral Partial Assignments and Separation}

We first present a useful lemma that allows us to consider only the valid assignments $\vec{g}$ that are integral, i.e., $\{0,1\}$-matrices. This lemma follows from the integrality of the $b$-matching polytope.

\begin{lemma}
\label{lem:bing}
For any $(\vec{x}, \vec{y})$, $\MFN(\vec{g}, \vec{x}, \vec{y})$ is feasible for all valid $\vec{g}$
if and only if $\MFN(\vec{\hat g}, \vec{x}, \vec{y})$ is feasible for all valid $\vec{\hat g}$  that
are integral.
\end{lemma}
\begin{proof}
  It is clear that if the flow network is feasible for all valid $\vec{g}$ then it is also feasible
  for the subset that are integral. We  show the harder side. Suppose $\MFN(\vec{\hat g}, \vec{x},
  \vec{y})$ is feasible for all valid $\vec{\hat g}$  that are integral and consider an arbitrary valid assignment $\vec{g}$
  that may be fractional. We will show that $\MFN(\vec{g}, \vec{x}, \vec{y})$ is feasible.

  Construct a complete bipartite graph with vertices $\cF \cup \cD$ and interpret $\vec{g}$ as the
 weights on the edges of this complete bipartite graph. As $\vec{g}$ is valid, we have $\sum_{j\in \cD}
  g_{ij} \leq U_i$ for each $i\in \cF$ and
  $\sum_{i\in \cF} g_{ij} \leq 1$ for each $j\in \cD$. In other words, $\vec{g}$ is a fractional
  solution to the $b$-matching polytope.  By the integrality of the $b$-matching polytope~(see e.g. \cite{S2003}),  we
  can write $\vec{g}$ as a convex combination of valid integral assignments $\vec{\hat g}^1, \vec{\hat
    g}^2, \ldots, \vec{\hat g}^r$, \ie there exist $\lambda_1, \lambda_2, \ldots, \lambda_r \geq
  0$ such that $\sum_{k=1}^r \lambda_k = 1$ and $\vec{g} = \sum_{k=1}^r \lambda_k\vec{\hat g}^k$.

  Now, let $\vec f^k$ denote the feasible flow for $\MFN(\vec{\hat g}^k,
  \vec{x},\vec{y})$, and choose $\vec f=\sum_{k} \lambda_k \vec f^k$. Observe that $\vec f$ is a feasible solution to $\MFN(\vec{g}, \vec{x}, \vec{y})$, since all the capacities and demands of $\MFN(\cdot,\vec x,\vec y)$ are given as linear functions of $\vec g$.
\end{proof}

A natural question is whether \MFNLP can be separated in polynomial time. While we currently do not
know if this is the case, we will establish in this paper that the feasibility constraint of
$\MFN(\vec { g},\vec x,\vec y)$ can be separated for any fixed $\vec{ g}$, and that this is
sufficient to find a fractional solution whose cost is within a constant factor from the optimum: in
a sense, this oracle enables us to extract the power of our strong relaxation within a constant
factor. The following lemma states the oracle. It follows from known characterizations using
LP-duality of multi-commodity flows and its proof can be found in Appendix~\ref{ap:lp}.

\begin{lemma}
\label{lem:linprog}
Given $\vec{ g^{\star}}$ in addition to $(\vec x^{\star},\vec y^{\star})$ such that $\MFN(\vec{ g^{\star}},\vec x^{\star},\vec y^{\star})$ is infeasible, we can find in polynomial time a violated inequality, i.e., an inequality that is valid for \MFNLP but violated by $(\vec x^{\star},\vec y^{\star})$. Moreover, the number of bits needed to represent each coefficient of this inequality is polynomial in $|\cF|$, $|\cD|$, and $\log U$, where $U:=\max_{i\in\cF}U_i$.
\end{lemma}

\subsection{\MFNLP is a Relaxation of the Capacitated Facility Location Problem}

We show in this subsection that \MFNLP is indeed a relaxation.

\begin{lemma}
\MFNLP is a relaxation of the capacitated facility location problem.
\end{lemma}

\begin{proof}
  Consider an arbitrary integral solution $(\vec{x}^{\star}, \vec{y}^{\star})$ to the facility location
  problem. By Lemma~\ref{lem:bing} we only need to verify that $\MFN(\vec{ g}, \vec{x}^{\star},
  \vec{y}^{\star})$ is feasible for each valid integral assignment $\vec{ g}$. Let
  $\vec{\hat g}$ be an arbitrary valid integral assignment.

Now we consider a directed bipartite graph $G=(V,A)$, of which one side of the vertex set is $\cD$, and on the other side, each facility $i\in\cF$ appears in $y^{\star}_i\cdot U_i$ duplicate copies. Consider the following two matchings $M_1$ and $M_2$ on these vertices.\begin{itemize}
\item For each client $j$, $M_1$ has an edge between $j$ and (a copy of) $i$ for which $x^{\star}_{ij}=1$. There will always be a copy of $i$ since $y_i^{\star}\geq x_{ij}^{\star}=1$. We will also ensure that a single copy of a facility does not have more than one incident edge: this is possible due to the capacity constraints on $(x^{\star},y^{\star})$.
\item For each $(i,j)$ such that $\hat g_{ij}=1$ and $y^{\star}_i=1$, $M_2$ has an edge between a copy of facility $i$ and client $j$. Note that no client will have more than one incident edge since $\sum_{i\in\cF}\hat g_{ij}\leq 1$. We will also ensure that a single copy of a facility does not have more than one incident edge. This is possible since $\sum_{j\in\cD}\hat g_{ij}\leq U_i$.
\end{itemize}Now we orient every edge in $M_1$ from clients to facilities; edges in $M_2$ are oriented in the opposite direction. $A$ is defined as the union of these two directed matchings. Since both $M_1$ and $M_2$ are matchings, every vertex in $G$ has indegree of at most one and outdegree of at most one. Hence, we can decompose $A$ into a set of maximal paths and cycles. Moreover, since $M_1$ matches every client, none of these maximal paths will end at a client. Reinterpret these paths as paths on $\cD$ and $\cF$, instead of on the duplicate copies of facilities. Let $\cP$ denote the set of these (nonempty) paths.

We will now construct a feasible multi-commodity flow on $\MFN(\vec{\hat g},\vec x^{\star},\vec y^{\star})$. We consider each $P\in\cP$. If $P$ starts from a facility, ignore it; otherwise let $j$ be the starting point of $P$ and $i$ the ending point: $P=( j,i_1,j_2,i_2, \ldots,j_k, i)$. If $d_j=0$, we ignore $P$; otherwise, we push one unit of flow of commodity $j$ along $P$, staying within the shaded area of Figure~\ref{fig:MFN}: i.e., the flow is pushed along $( j^s,i_1,j_2^s,i_2, \ldots,j_k^s, i)$. When we arrive at $i$, further push this flow along $( i,i',j^t) $, draining the flow at $j^t$: this is legal since the flow is of commodity $j$. We repeat this until we have considered all paths in $\cP$. We claim that this procedure yields a feasible multi-commodity flow.

First, note that each arc in $A$ maps to an edge of capacity 1 in $\MFN(\vec{\hat g},\vec
x^{\star},\vec y^{\star})$.
Since $\cP$ is a decomposition of
(a subset of) $A$, capacity constraints on $( j^s,i)$ and $( i,j^s)$ are
satisfied from the construction. Now consider the capacity of $( i,i')$. Each time we
encounter a path $P\in\cP$ that starts at some client and ends at $i$, one unit of additional flow
is sent over this arc. If $y^{\star}_i=0$, there will be no such path in $\cP$. If $y^{\star}_i=1$,
there are at most $U_i-\sum_{j\in\cD}\hat g_{ij}$ paths in $\cP$ ending at $i$, since $M_2$ matches
exactly $\sum_{j\in\cD}\hat g_{ij}$ copies of $i$ out of $U_i$ in total. This verifies that the
capacity constraint on $( i,i')$ is also satisfied. Finally, arc $( i',j^t)$
is used only when we process $P\in\cP$ that starts from $j$ and ends at $i$. This is true for at
most one path in $\cP$ since there is at most one path starting from each client (note that there
are no duplicate copies of clients in $G$); moreover, $P$ can end at $i$ only if $y^{\star}_i=1$
(otherwise, there are no copies of $i$ in $G$). The capacity constraint on $( i',j^t)$
is therefore also satisfied.

Demand constraints are also satisfied: suppose $d_j=1$ for some $j\in\cD$. This means $\vec{\hat g}$ does not assign $j$ to any facility, and therefore $M_2$ does not match $j$. Hence $j$ has indegree of zero and outdegree of one in $G$, and thus $\cP$ contains exactly one path that starts from $j$.
\end{proof}

Intuitively, the above proof can also be interpreted as follows: given an arbitrary partial assignment and integral solution, consider the shaded area of Figure~\ref{fig:MFN}. By saturating every arc in this area, we obtain a feasible single-commodity flow where every client generates a unit flow either at its original position or at the facility it is assigned to by $\vec g$. While this flow satisifies every commodity-oblivious capacity, it may not be immediately clear why it also satisfies the commodity-specific capacities; here we can appeal to the integrality of $\vec y^\star$, because in this case every facility with nonzero commodity-oblivious capacity will automatically have the full commodity-specific capacity of $1$. Such an argument, however, would not extend to a fractional solution (to the standard LP for example), which illustrates the strength of our relaxation.

\subsection{Comparing \MFNLP to Standard LP and  Knapsack-Cover Inequalities}\label{sec:knapsack}
In order to facilitate our understanding of the new relaxation, we demonstrate how it relates to
other formulations for the capacitated facility location problem.

\paragraph{Standard LP.} We shall show that the constraint that $\MFN(\vec g,\vec x,\vec y)$ is feasible for $\vec g=\vec
0$ already is sufficient to see that our relaxation is no worse than the standard LP relaxation. The
standard LP relaxation uses the same variables $(\vec{x},\vec{y})$ as \MFNLP and is formulated as
follows:
\begin{center}
\boxed{%
  \addtolength{\linewidth}{-2\fboxsep}%
  \addtolength{\linewidth}{-2\fboxrule}%
        \begin{minipage}{10cm}%
       \begin{align}
&&&\textrm{minimize} \quad  \rlap{$\displaystyle\sum_{i\in \cI} o_i \cdot y_i + \sum_{i\in \cF, j\in \cD} c_{ij} \cdot x_{ij}$,}\nonumber\\[2mm]
       &&&     \textrm{subject to} & x_{ij}&\leq   y_i & \forall i\in \cF, j\in \cD;\label{e:standard:1}\\[1mm]
               &&& &\sum_{i\in \cF} x_{ij}&= 1 & \forall  j\in \cD;\label{e:standard:2}\\[1mm]
             &&&  & \sum_{j\in \cD} x_{ij}&\leq y_iU_i & \forall  i\in \cF;\label{e:standard:3}\\[1mm]
          &&& & 0 \leq \vec{x},&\vec{y}\leq 1
          \end{align}
        \end{minipage}%
}
\end{center}

Consider arbitrary $(\vec x^{\star},\vec y^{\star})$ that makes $\MFN(\vec 0,\vec
x^{\star},\vec y^{\star})$ feasible. Since $\vec g=\vec 0$, the support of $\MFN(\vec g,\vec
x^{\star},\vec y^{\star})$ is acyclic and the flow of each commodity can be decomposed into paths
with no cycles. In particular, every path for commodity $j\in\cD$ will be in the form of
$j^s-i-i'-j^t$ for some $i\in\cF$. Observe that this implies that the only commodity that has
nonzero flow on $( j^s,i)$ is $j$; let this flow be $\bar x_{ij}$. Now we claim that
$(\vec {\bar x},\vec y^{\star})$ is a feasible solution to the standard LP: \eqref{e:standard:2}
follows from $d_j=1$; \eqref{e:standard:1} follows from the capacity constraint on $(
i',j^t)$. Finally, $\sum_{j\in\cD}\bar x_{ij}$ equals the total (regardless the commodity)
incoming flow to $i$; this in turn is bounded from above by $y^{\star}_i U_i$ from the capacity
constraint on $( i,i')$. This shows that $(\vec {\bar x},\vec y^{\star})$ is feasible to
the standard assignment LP. Observe that $\vec x^{\star}$ dominates $\vec{\bar x}$ and therefore the
lower bound on the optimum given by \MFNLP is always no worse than the standard assignment LP.

\paragraph{Knapsack-cover inequalities.}
Consider a special case where the metric on $\cF$ and $\cD$ is constantly zero: i.e., every facility
and every client are ``on the same spot''. As there is no connection cost, the problem reduces to
simply selecting a set of facilities that as a whole has enough capacity while minimizing the total
cost: this is the minimum knapsack problem. Each facility $i$ corresponds to an item with weight $U_i$ and cost $o_i$; $|\cD|$
corresponds to the demand of the knapsack problem. Using the notation of the capacitated facility
location problem, the knapsack cover inequalities due to Carr et al.~\cite{CFLP2000} are written as follows:\begin{equation*} \sum_{i\in\cF\setminus A} \min
  (U_i,|\cD|-{\textstyle\sum_{i\in A}U_i})\cdot y_i\geq |\cD|-{\textstyle\sum_{i\in A}U_i},\quad
  \forall A\subset\cF\textrm{ such that }{\textstyle\sum_{i\in A}U_i}\leq |\cD| .\end{equation*}

Now each of these inequalities are implied by our relaxation. Let $S$ be a set of any
${\textstyle\sum_{i\in A}U_i}$ clients and $R:=\cD\setminus S$. Consider $g$ that fully assigns
every client in $S$ to the facilities in $A$, thereby saturating those facilities, and does not
assign any clients in $R$. Note that $d_j$ will be zero for every $j\in S$ and one for every $j\in
R$. Now we choose a feasible solution $(\vec{z},\vec{\ell})$ to \eqref{e:seplp:1}-\eqref{e:seplp:2} as follows: for
each facility $i\in A$, if $U_i> |\cD|-{\textstyle\sum_{i\in A}U_i}$, set $l_{(
  i',j^t)}:=1$ for all $j\in R$; if $U_i\leq |\cD|-{\textstyle\sum_{i\in A}U_i}$, set
$\ell_{( i,i')}:=1$. All other $\ell$'s are set to zero. $z_j:=1$ for all $j\in R$; $z_j:=0$
for all $j\in S$. Now \eqref{e:mfdual} implies the knapsack cover
inequalities.

\paragraph{Example.}

We give a simple integrality gap example for the standard LP and show how our linear program strengthens the linear program to \emph{cut off} the fractional solution.
Consider the following instance of the capacitated facility location problem. Here we have two facilities $i_1$ and $i_2$ each with capacity $n$ and opening costs are $0$ and $1$ respectively. There are $n+1$ clients $j_1,\ldots,j_{n+1}$. The distance between any two points, either facility or client, is zero. Thus all facilities and clients sit at the same point. Consider the following fractional solution $(\vx^*,\vy^*)$ where we have $y^{\star}_{i_1}=1$ and $y^{\star}_{i_2}=\frac{1}{n}$,  $x^{\star}_{i_1j_r}=\frac{n}{n+1}$ for each $1\leq r \leq n+1$ and $x^{\star}_{i_2j_r}=\frac{1}{n+1}$ for each $1\leq r \leq n+1$. It is quite simple to verify that $(\vx^{\star},\vy^{\star})$ is a feasible solution to the standard LP and costs $\frac{1}{n+1}$ while the cost of the optimal solution is $1$ giving us an unbounded integrality gap for large $n$. We now show how this fractional solution can be \emph{cut off} using our stronger LP. We also note that, for this instance, the knapsack cover inequalities are enough to obtain a good approximation.

Consider the partial assignment $\vg^{\star}$ defined as follows. We let $g^{\star}_{i_1j_r}=1$ for
each $1\leq r\leq n$. Thus we have assigned $n$ clients to facility $i_1$ and saturating it. The
only facility that can serve the unassigned client is $i_2$. Consider the flow network defined by
this instance. The capacity of arc $(i_2',j_{n+1}^t)$ is $y_{i_2}$ and demand of client $j_{n+1}$ is
$1-\sum_{i\in \cF} g^{\star}_{ij_{n+1}}=1$. Since all flow reaching $j^t$ must go on this arc, we
must have $y_{i_2}\geq 1$. Thus the fractional solution must cost at least one.


\section{Approximation Algorithm}\label{sec:algorithm}

\newcommand{\brx}{\bar{\vx}}
\newcommand{\bry}{\bar{\vy}}

In this section, we describe our approximation algorithm and prove Theorem~\ref{thm:main}:
\footnote{The cost function includes two components, facility opening costs and connection costs. Optimizing the parameters to obtain the same worst case performance for both components will lead to significant improvements in the constant obtained above. But such methods will not lead to improvement over $5$-approximation due to local search~\cite{BGG12}.}
\begin{reptheorem}{thm:main}[restated]
There exists a 288-approximation algorithm for the capacitated facility location problem. The cost of its output is no more than 288 times the optimal cost of \MFNLP.
\end{reptheorem}

 The
algorithm is based on rounding a given fractional ``solution'' to \MFNLP. However, as we do not know how to
solve \MFNLP exactly, we give a \emph{relaxed} separation oracle
that either outputs a violated
inequality or returns an integral solution obtained from the fractional solution by increasing the
cost only by a constant factor. A similar approach has previously been used by Carr et al.~\cite{CFLP2000} and later by Levi et al.~\cite{LeviLS07}.

\paragraph{Algorithm overview.} Our algorithm first guesses the cost of the optimal solution to \MFNLP using a binary
search\footnote{We remark that the relaxed separation oracle can also  simply be used with the standard optimization version of the ellipsoid method, which would not involve a binary search.}.
For each guess, say $\gamma$, we run an ellipsoid algorithm. At each step of the ellipsoid
algorithm, we obtain a fractional solution $(\vx^{\star},\vy^{\star})$, possibly infeasible. We then
first verify the boundary constraints $\vec{0} \leq \vx^\star, \vy^\star \leq \vec{1}$ and the objective
constraint $c(\vx^\star, \vy^\star) \leq \gamma$. If $(\vx^\star, \vy^\star)$ violates one of these
inequalities, we output it and continue to the next iteration of the ellipsoid algorithm. Otherwise,
we either construct  a so-called semi-integral solution (defined below) or output a violated
inequality showing  infeasibility of the flow network $\MFN(\vec{g}^\star, \vx^\star, \vy^\star)$ for some $\vec{g}^\star$. In the final step, our algorithm rounds this semi-integral
solution into an integral solution by increasing the cost by a constant factor.

We remark that the main step of our algorithm exploiting the strength of \MFNLP is
the step for finding a semi-integral solution or outputting a violated inequality (summarized in
Theorem~\ref{thm:mainalg2}). An interesting detail is that our rounding algorithm only needs that
the multi-commodity flow network is feasible for a \emph{single} $\vec{g}^\star$ in order to output a
semi-integral solution.  Once we have a semi-integral solution, the rounding is fairly
straightforward using previous algorithms for soft-capacitated versions. We now first define
semi-integral solutions and describe the rounding to integral solutions in
Section~\ref{sec:defsemi}. We then continue with the proof of Theorem~\ref{thm:mainalg2} which is the
main technical contribution of this section.

\subsection{Semi-Integral Solutions: Definition and Rounding}
\label{sec:defsemi}
The idea of semi-integral solutions is that they partition the
facilities into two sets: the set $I$ of integrally opened facilities  and the set $S$ of
facilities of small opening. Clients may be fractionally assigned to both facilities in $I$ and $S$.
However, there is an important constraint regarding the assignment to facilities in
$S$ (condition \eqref{cond:semi3} in the definition below).  For each client $j$, it says that at
most a $y_i$ fraction of $j$'s total assignment to facilities in $S$ can be assigned to $i\in S$.
This will allow us to round semi-integral solutions by using techniques developed for the standard
LP relaxation.

\begin{definition}\label{def:semi}
A solution $(\hat \vx,\hat \vy)$ is \emph{semi-integral} if it satisfies the following conditions.
\begin{enumerate}[(i)]
\item $(\hat\vx,\hat\vy)$ satisfies the assignment constraints, i.e., for each $j\in \cD$, $\sum_{i\in \cF} \hat x_{ij}=1$ and for each $i\in \cF$, $\sum_{j\in \cD}\hat x_{ij}\leq \hat y_i U_i.$ \label{cond:semi1}
\item For each $i\in \cF$, $\hat y_i=1$ or $\hat y_i\leq \frac12$. Let $I=\{i:\hat y_i=1\}$ and $S=\cF\setminus I$.\label{cond:semi2}
\item For each $j\in \cD$, let $\hat d_j=\sum_{i\in S} \hat x_{ij}$. Then we have $\hat x_{ij}\leq \hat y_i \hat d_j$ for each $i\in S$ and $j\in \cD$.\label{cond:semi3}
\end{enumerate}
\end{definition}

We now describe the  procedure for rounding the semi-integral solution to an integral
solution. All facilities in $I$, whose opening variables are set to one in the
semi-integral instance, are opened. Consider the residual instance where each client has a residual
demand $\hat d_j$, amount to which it is not assigned to facilities in $I$. This residual demand is
satisfied by facilities in $S$, each of which is open to a fraction of at most $\frac12$ by the
semi-integral solution. Conditions~\eqref{cond:semi1} and \eqref{cond:semi3} of the semi-integral
solution imply that the residual solution is a feasible solution to the standard LP for the residual
instance.
Since the opening variables are set to a small fraction in the residual instance, we can use
an approximation algorithm for the soft-capacitated facility location problem which rounds the
standard LP. An $(\alpha,\beta)$-approximation algorithm for the soft-capacitated facility location
problem returns a solution whose cost is no more than $\alpha$ times the cost of the optimal
fractional solution and violates the capacity of any open facility by a factor of at most $\beta$. We give
the algorithm our residual instance as input where we scale down the capacities by a factor of
$\beta$ but scale up the opening variables by the same
factor.
Observe that as long as each $\hat{y}_i\leq \frac{1}{\beta}$ for each facility $i\in S$, we obtain a
feasible solution to the standard LP.
Here we use the
$(18,2)$-bicriteria approximation algorithm due to Abrams et al.~\cite{AMMP2002} to complete our
rounding to an integral solution.

\begin{lemma}\label{lemma:main3}
 Given a semi-integral solution $(\hat{\vx},\hat{\vy})$, we can in polynomial time find  an integral solution $(\bar{\vx},\bar{\vy})$ whose cost is at most $36 c(\hat{\vx},\hat{\vy})$.
\end{lemma}

We give the formal proof of Lemma~\ref{lemma:main3} in Appendix~\ref{sec:soft}.

\subsection{Finding a  Semi-Integral Solution or a Violated Inequality}\label{sec:partial}
We are now ready to describe and prove the main ingredient of our rounding algorithm.
\begin{theorem}\label{thm:mainalg2}
There is a polynomial time algorithm that, given $(\vx^{\star},\vy^{\star})$, either
\begin{itemize}
\item shows that $(\vx^{\star},\vy^{\star})$ is infeasible for $\MFNLP$ and returns a violating inequality, or
\item returns a solution $(\hat{\vx},\hat{\vy})$ such that $(\hat{\vx},\hat{\vy})$ is semi-integral and $c(\hat{\vx},\hat{\vy})\leq 8 c(\vx^{\star},\vy^{\star})$.
\end{itemize}
\end{theorem}
Note that the above theorem together with Lemma~\ref{lemma:main3}  implies Theorem~\ref{thm:main}
with the claimed approximation guarantee $8\cdot 36 = 288$.

We prove the theorem by describing the algorithm together with its properties. For an overview of the algorithm, see also Figure~\ref{fig:partial}.
The algorithm
consists of several steps.
First, we round up
the large opening variables of $\vy^\star$ to obtain modified opening variables $\vy'$. We define\[
y'_i:=\begin{cases} 
  1,&\textrm{if }y^\star_i\geq\frac{1}{4};\\
  y_i^{\star},&\textrm{otherwise}.
\end{cases}\] Let $I$ be the set of facilities that are fully open by $\vy'$: $I:=\{i\in\cF:
y'_i=1\}$. $S$ denotes the set of facilities that are open by a small fraction: $S:=\cF\setminus I$.

Given that our algorithm is going to open all the facilities in $I$, we will try to find a partial
assignment $\vec g^{\star}$ that
assigns as many clients to
these facilities as possible, while at the same time ensuring $\vec g^\star$ does not become too costly compared to $\vx^*$. To this end, we will
derive $\vec g^\star$ from a maximum $b$-matching in a bipartite graph on $\cF$ and $\cD$ whose edges are capacitated by $2\vx^*$. Let $G=(\cD,I,E)$ be the
complete bipartite graph whose bipartition is given by the clients  $\cD$ and the opened facilities $I$. An arc $(j,i)$ where $j\in
\cD$ and $i\in I$ is given a capacity of $2x^{\star}_{ij}$. This is to ensure that the cost of the matching is within a factor of 2 compared to the original assignment cost.  Every client $j$ has a capacity of one
and each facility $i\in I$ is given a capacity of $U_i$. Let $\vec z$ denote a maximum fractional
$b$-matching of $G$. Note that the matching may not be integral because of the capacities on the
edges. As $\vec{z}$ is a maximum fractional matching, its residual network  $H$ with arc set
$\{(j,i) : z_{ij} < 2x^\star_{ij}\} \cup \{(i,j):  z_{ij} > 0\}$ has useful properties that we describe below. In
particular, if we consider an \emph{unsaturated} client $j$, i.e., $\sum_{i\in I} z_{ij} < 1$, then $j$ has
no path in $H$ to a facility $i$ with remaining capacity, as that would contradict that $\vec{z}$ is
a maximum matching.

We shall now formalize these properties. Let us call a client $j\in\cD$ \emph{saturated} if $\sum_{i\in\cF}z_{ij}=1$, and \emph{unsaturated} otherwise; define
\begin{align*}
I_H & :=  \{i \in I: i \mbox{ is reachable in $H$ from some client $k$ that was unsaturated}\}; \\
D_H & := \{j \in \cD: j \mbox{ is reachable in $H$ from some client $k$ that was unsaturated}\}.
\end{align*}
Similar to clients, a facility $i\in I$ is called \emph{saturated} if $\sum_{j\in \cD} z_{ij}=U_i$ and \emph{unsaturated} otherwise. The following lemma summarizes three useful observations on $\vec{z}$ and $H$.
\begin{lemma}\label{lem:max-flow-properties}
The following must hold.
\begin{enumerate}[(a)]
\item \label{flow:prop1} Any facility $i\in I_H$ is saturated, i.e., $\sum_{j\in \cD} z_{ij} = U_i$.
\item \label{flow:prop2} If $i\in I\setminus I_H$ and $j\in D_H$, $z_{ij} = 2x^{\star}_{ij}$.
\item \label{flow:prop3} If $i\in I_H$ and $j\in \cD\setminus D_H$, $z_{ij}=0$.
\end{enumerate}
\end{lemma}
\begin{proof}
We first prove \eqref{flow:prop1}. Suppose toward contradiction that there exists a facility $i\in I_H$ that is not saturated.  By the
definition of $I_H$ there exists a client $k$ that is unsaturated and $i$ is reachable from $k$
in $H$. Therefore there exists an alternating path from $k$ to $i$ which contradicts that the chosen
fractional matching $\vec{z}$ was maximum.

We now prove \eqref{flow:prop2}. By the definition of $D_H$, there exists an unsaturated client $k$ such that $j$ is reachable
from $k$ in $H$. Therefore, any facility $i$ such that $z_{ij} < 2x^{\star}_{ij}$ is also reachable from
$k$ and therefore part of $I_H$.
The proof of \eqref{flow:prop3} follows from the fact that $(i,j)\notin H$ since $i$ is reachable from an unsaturated client and $j$ is not. Therefore, $z_{ij}=0$.
\end{proof}

Now the valid partial assignment $\vec{g}^\star$ is constructed as follows:
\begin{align}
\label{eq:defg}
g^{\star}_{ij} =
\begin{cases}
z_{ij}&  \mbox{if } i \in I_H\\
z_{ij}& \mbox{if } i\in I\setminus I_H, j\in \cD\setminus D_H\\
0& \mbox{if }i\in I\setminus I_H, j\in  D_H\\
0& \mbox{if }i\in S.
\end{cases}
\end{align}
Note that $\vec{g}^\star$  is defined in terms of $\vec{z}$. This will allow us to analyze the flow
network using the properties of $\vec z$ described in Lemma~\ref{lem:max-flow-properties}.

Once we have this partial assignment, the algorithm verifies if
$\MFN(\vg^{\star},\vx^{\star},\vy^{\star})$ is feasible.  If not, we
invoke Lemma~\ref{lem:linprog} to find a violated inequality and Theorem~\ref{thm:mainalg2} holds.
Otherwise, the algorithm proceeds to construct a semi-integral solution using this partial assignment. For the rest of this section, we will assume that $\MFN(\vec g^{\star},\vec x^{\star},\vec y^{\star})$ is
feasible. Note that the feasibility of $\MFN(\vec g^{\star},\vec x^{\star},\vec y^{\star})$
guarantees the feasibility of $\MFN(\vec g^{\star},\vec x^{\star},\vec y')$ since $\vec y'\geq \vec
y^{\star}$.
\begin{claim}\label{lem:increasey}
If $\MFN(\vec g^{\star},\vec x^{\star},\vec y^{\star})$ is feasible and $\vec y'\geq \vec y^{\star}$, $\MFN(\vec g^{\star},\vec x^{\star},\vec y')$ is feasible.
\end{claim}
\begin{proof}
Consider a feasible flow for $\MFN(\vec g^{\star},\vec x^{\star},\vec y^{\star})$. Observe that it is feasible for $\MFN(\vec g^{\star},\vec x^{\star},\vec y')$ as well, since the arc capacities of $\MFN(\vec g^{\star},\vec x^{\star},\vec y)$ is nondecreasing in $\vec y$ while the demands remain the same since they depend on $\vg^{\star}$.
\end{proof}

We have now made our choice of $\vg^{\star}$ that satisfies the following three key properties which help us round
$(\vx^{\star},\vy^{\star})$:\begin{enumerate}
\item $\vec g^{\star}\leq \vec z\leq 2\vec x^{\star}$ and therefore $c(\vec{g}^\star)\leq 2c(\vec{x}^{\star})$;
\item $\vec g^{\star}$ assigns clients only to the fully open facilities, i.e., facilities in $I$;
\item $\vec g^{\star}$ satisfies the property formalized by Lemma~\ref{lem:flow}. (Note that Lemma~\ref{lem:flow} is proven for our carefully constructed partial assignment. It does not hold in general for arbitrary partial assignments.)
\end{enumerate}
Let $\vec{f}$ denote the flow certifying the feasibility of
$\MFN(\vg^{\star},\vx^{\star},{\vy}')$. We decompose $\vec{f}$ into flow paths where we let
$\cP_{ij}$ denote the set of flow paths carrying non-zero flow from $j^s$ to $j^t$ that use the arc
$(i,i')$. That is, these are the paths which take flow from $j$ and sink it at $i$. Let $f(P)$ denote
the flow on a path $P\in \cP_{ij}$. For each $i\in \cF$ and $j\in \cD$, we let $h(i,j)=\sum_{P\in
  \cP_{ij}}f(P)$ denote the amount of flow that client $j$ sinks at facility $i$. For any subset
$X\subseteq \cF$ and $j\in \cD$, let also $h(X,j):= \sum_{i\in X} h(i,j)$, i.e., the total amount of
flow that client $j$ sinks at facilities in $X$.

\begin{lemma}\label{lem:flow}
There exists a feasible flow to the multi-commodity flow problem $\MFN(\vg^{\star},\vx^{\star},\vy')$ such that each client $j \in \cD$ sends at least half its demand to facilities in $S$, i.e., $h(S,j)\geq \frac{d_j}{2}=\frac{1}{2}(1-\sum_{i\in \cF} g^{\star}_{ij})$.
\end{lemma}
\noindent The proof of this lemma can be found in Section~\ref{sec:niceflow}.

Observe that the flow satisfying the conditions in Lemma~\ref{lem:flow} can be obtained in polynomial time by adding
additional linear constraints to the multi-commodity flow linear program for
$\MFN(\vg^{\star},\vx^{\star},{\vy}')$.  Let $\vec{f}$ denote such a flow. The algorithm now
proceeds by using this flow to define a semi-integral solution $(\hat \vx, \hat \vy)$.
Lemma~\ref{lem:flow} guarantees $h(S,j) \geq d_j/2$; hence we define the semi-integral solution by scaling up this assignment
by a factor of at most $2$. This ensures that  each client assigns all its  demand $d_j$ to $S$ and
that it is a semi-integral solution. Formally, we construct the semi-integral solution
$(\hat{\vx},\hat{\vy})$ as follows:

$$\hat{y}_i =
\begin{cases}
1,&  \mbox{if } i \in I;\\
2y^{\star}_i,& \mbox{if }i\in S;
\end{cases}
 ~~~~ \qquad \qquad \qquad~~~~
\hat{x}_{ij} =
\begin{cases}
g^{\star}_{ij},&  \mbox{if } i \in I, j\in \cD;\\
d_j\frac{ h(i,j)}{h(S,j)} ,&\mbox{if } i\in S, j\in \cD;
\end{cases}
$$
where we define $\frac{h(i,j)}{h(S,j)}$ to be $0$ if $h(S,j) = 0$.
For $i\in S$ and $j\in \cD$, we have $\hat{x}_{ij} =h(i,j) \cdot \frac{d_j }{h(S,j)} \leq 2 h(i,j)$ from
Lemma~\ref{lem:flow}. 
This allows us to bound the total cost of solution $(\hat{\vx},\hat{\vy})$.

\begin{lemma}\label{lem:cost}
The solution $(\hat{\vx},\hat{\vy})$ is semi-integral and $c(\hat{\vx},\hat{\vy})\leq 8 c(\vx^{\star},\vy^{\star})$.
\end{lemma}
The above lemma finishes the proof of Theorem~\ref{thm:mainalg2}. Its proof is
fairly straightforward calculations and can be found in Section~\ref{sec:rounding}.


\subsubsection{Proof of Lemma~\ref{lem:flow}: Existence of a Nice Multi-Commodity Flow}\label{sec:niceflow}
In this section, we prove Lemma~\ref{lem:flow}, i.e., we show there is  always a flow $\vec{f}$ as a
solution to $\MFN(\vg^{\star},\vx^{\star},\vy')$ satisfying the conditions of the lemma.

By definition, any client $j\in\cD\setminus D_H$ is saturated. Moreover,
$\sum_{i\in\cF}g^{\star}_{ij}=\sum_{i\in I}z_{ij}=1$. Thus, every client $j$ whose demand is
non-zero in $\MFN(\vec g^{\star},\vec x^{\star},\vec y')$, i.e.,
$d_j:=1-\sum_{i\in\cF}g^{\star}_{ij}>0$, is in $D_H$. On the other hand, for each $i\in I_H$ we have
$\sum_{j\in \cD} g^{\star}_{ij}=\sum_{j\in \cD} z_{ij}=U_i$ where the last equality follows from
property~\eqref{flow:prop1} of Lemma~\ref{lem:max-flow-properties}. Thus the commodity-oblivious
capacity of $i$, i.e., the arc capacity of $(i,i')$, is $y'_i(U_i-\sum_{j\in \cD}
g^{\star}_{ij})=0$. In summary, every client $j$ with nonzero demand is in $D_H$, and drains its flow at
facilities in $I\setminus I_H$ or at facilities in $S$. Now for any $j\in D_H$, we have
\begin{equation}\label{e:dhhalf}
d_j= 1-\sum_{i\in \cF} g^{\star}_{ij}=1-\sum_{i\in I_H} z_{ij}\geq \sum_{i\in I\setminus I_H} z_{ij}=\sum_{i\in I\setminus I_H}2 x^{\star}_{ij}
\end{equation}
where the second equality follows from the definition of $g^{\star}$, the inequality follows
from the fact that $\vec z$ is a fractional $b$-matching satisfying capacity one at $j$ and last equality
follows from Condition~\eqref{flow:prop2} of Lemma~\ref{lem:max-flow-properties}.

\newcommand{\thedefn}{\eqref{eq:defg}}

Consider a feasible flow $\vec f$ to $\MFN(\vg^{\star},\vx^{\star},\vy')$, decomposed into paths and cycles. We will call these paths and cycles \emph{flow paths} and \emph{flow cycles}, respectively. Without loss of generality, we can assume that there exist no flow cycles, and every flow path sends nonzero flow on it. The following claim simplifies our proof by letting us ignore the less interesting part of the flow network.

\begin{claim}\label{cl:flow1}
We can assume without loss of generality that no flow path contains node $k^s$ for $k\in\cD\setminus D_H$.
\end{claim}
\begin{proof}
For an arbitrary client $k$ in $\cD\setminus D_H$, consider node $k^s$ in the flow network. Among its incoming arcs, every arc $(i,k^s)$ coming from $i\in I_H$ has zero capacity, since $g^{\star}_{ij}=0$ from Property~\eqref{flow:prop3} of Lemma~\ref{lem:max-flow-properties}. Arcs from $i\in S$ also have zero capacities (see \thedefn). Thus, the only incoming arcs to $k^s$ that have nonzero capacity are arcs $(i,k^s)$ where $i\in I\setminus I_H$.

Now we examine each flow path one by one and modify them, thereby defining a new flow. Suppose a
flow path $P$ starts at $j^s$ and ends at $j^t$. We truncate this path at the first facility in
$I\setminus I_H$ on the path, say $i_0$, and then send the flow directly to the sink $j^t$ by
appending $(i_0)-i_0'-j^t$ at the end. If the path does not contain any $i\in I\setminus I_H$, we do
nothing. We are making no other changes to the flow paths, including the commodity and the amount of
flow on them. Once this modification has been applied to all the paths, no flow path uses any arc of
form $(i,k^s)$ where $i\in I\setminus I_H$; the flow paths, therefore, cannot encounter any $k^s$
such that $k\in \cD\setminus D_H$. Recall that $k\in \cD\setminus D_H$ has zero demand and hence
$k^s$ cannot become the first (or last for $k^t$) node in a flow path.
  
  We claim that this
  modification maintains feasibility. Observe that the demand of each client is still satisfied
  since we have only rerouted the flow on a different path from $j^s$ to $j^t$. The  arcs on
  which the flow may have increased are the arcs $(i,i')$ where $i\in I\setminus I_H$ and the arcs $(i',
  j^t)$ where $i \in I\setminus I_H$ and $j\in D_H$. It is easy to verify that an arc of the latter type, say $(i', j^t)$, has
   its capacity constraint satisfied: the arc's capacity is $y'_i d_j = d_j$ and client $j$ pushes at most $d_j$ units of
  flow. Let us now consider an arc of the first type, $(i,i')$ for $i\in I \setminus I_H$. Note that, after the truncations, the only incoming
  arcs it receives flow on are the ones from $j\in D_H$. The total incoming
  flow to $i$ can therefore be bounded by the total capacity of these arcs, which is
\begin{align}
\sum_{j\in D_H} x^{\star}_{ij}& \leq \sum_{j\in D_H} 2x^{\star}_{ij}=\sum_{j\in D_H} z^{\star}_{ij}\leq U_i-\sum_{j\in\cD\setminus D_H} z^{\star}_{ij}= U_i-\sum_{j\in\cD\setminus D_H} g^{\star}_{ij}
,\end{align}
where we use the fact that $z^{\star}_{ij}=2x^{\star}_{ij}$ if $i\in I\setminus I_H$ and $j\in D_H$
from Lemma~\ref{lem:max-flow-properties} for the first equality, the fact that $\vec{z}$ is a fractional matching with bound $U_i$ at
facility $i$ for the second inequality, and for the last equality $g^{\star}_{ij}=z^{\star}_{ij}$ if $i\in I\setminus I_H$ and $j\in \cD\setminus
D_H$. The feasibility of the truncated flow now follows from that  the capacity of arc $(i,i')$ is $y_i'(U_i-\sum_{j\in \cD} g^{\star}_{ij})=U_i-\sum_{j\in\cD\setminus D_H} g^{\star}_{ij}$.
\end{proof}

We  proceed to prove that there is a flow of $\MFN(\vec g^\star, \vx^\star, \vy')$ satisfying the
properties of Lemma~\ref{lem:flow}.

If no client has positive demand, there is nothing to prove. For any $X\subseteq \cD$, let $d(X):=\sum_{j\in X} d_j$ denote the total demand of clients in $X$. Consider $J \subseteq \cD$ defined as the set of
clients that have nonzero demand and send the smallest fraction of their demand to sinks in $S$. Formally,
let
$$
\alpha = \min_{j:d_j>0}h(S,j)/d_j \qquad \mbox{and} \qquad J = \{j \in \cD: h(S,j)/d_j = \alpha, d_j>0
\}.
$$
 We shall show that if $\alpha < 1/2$ then we can modify $\vec{f}$ so that we
  either decrease the cardinality of $J$ or increase $\alpha$. A flow  that minimizes the
  lexicographic order of $(-\alpha, |J|)$ must therefore have $\alpha \geq 1/2$ which proves the
  lemma. We remark that here we are only considering flows consistent  with
  Claim~\ref{cl:flow1}, i.e., no flow path passes through a node $k^s$ with $k\in \cD
  \setminus D_H$. Also note that  a minimum $\alpha$ exists because $\alpha$ is a continuous function of flow, and the set of feasible flows for $\MFN(\vec g^{\star},\vec x^{\star},\vec y')$ is compact.

Now suppose that $\alpha < 1/2$. Note that $J\subset D_H$ since $d_j>0$ implies $j\in D_H$. Then, by
the definition of $J$,
$$
\sum_{j\in J} h(I\setminus I_H,j) = (1-\alpha)d(J) > d(J)/2 \geq \sum_{i\in I\setminus I_H , j\in J} x_{ij}^\star,
$$
where the last inequality follows from~\eqref{e:dhhalf}.
In other words, the total flow $\sum_{ j\in J} h(I\setminus I_H ,j)$ from clients in $J$ to
facilities in $I \setminus I_H$  is
strictly greater than the sum of capacities $\sum_{i\in I\setminus I_H, j\in J}
x^\star_{ij}$ of the arcs from clients in $J$ to facilities in $I \setminus I_H$.

Therefore, not all the flow paths originating from $J$ can enter $I\setminus I_H$ directly from
$J$. That is, for some $j\in J$, $k\notin J$, and $i\in I\setminus I_H$, there exists a flow path
$P$ that starts from $j^s$ but enters $i$ via $k^s$. Let $P_1$ denote the subpath between $j^s$ and
$k^s$, and we have $P=j^s-P_1-k^s-i-i'-j^t$. 

By  Claim~\ref{cl:flow1}, we have $k\in D_H\setminus J$
and the demand of $d_k$ is positive by~\eqref{e:dhhalf} and the fact that $x_{ik} >0$ since $P$
carried non-zero flow. 
This together with $k\notin J$, implies that $k$ sends strictly more than $\alpha$ fraction of its demand to sinks in $S$. Let $v_1, v_2, \ldots, v_\ell$ be the sinks in $S$
to which $k$ sends nonzero flows. Since $k$ can send at most $d_{k} y'_v$ flow to a facility $v\in S$, we have
\[
\sum_{i=1}^\ell d_{k} y'_{v_i} > \alpha{d_{k}}
\implies \sum_{i=1}^\ell  y'_{v_i} > \alpha.
\] Moreover, as $j$ sends $\alpha d_j$ units of flow to $S$, it sends at most $\alpha d_j$ units to $\{v_1,\ldots,v_\ell\}\subset S$; therefore, there exists some sink $v\in
\{v_1, \dots, v_\ell\}$ to which $j$ sends strictly less than $y'_vd_j$ units of flow. Let $Q$ be a flow path along which $k$ sends flow to $v$, written as $Q=k^s-Q_1-v-v'-k^t$ for some subpath $Q_1$.

Now we ``exchange'' the suffixes of $P$ and $Q$ by a small amount. To be precise, we choose a sufficiently small $\epsilon>0$, and decrease the flow on $P$ and $Q$ each by $\epsilon$, but increase the flow on $P'=j^s-P_1-k^s-Q_1-v-v'-j^t$ and $Q'=k^s-i-i'-k^t$ by $\epsilon$. $P'$ and $Q'$ will be introduced as new flow paths if they do not already exist. Note that this satisfies the demand constraints, since we have only changed the intermediate vertices between the same pair of source and sink. We further claim that we can choose a positive $\epsilon$ that satisfies the capacity constraints. Observe that every arc in $P'$ and $Q'$ was used in $P$ and $Q$, except for $(v',j^t)$ and $(i',k^t)$. For other arcs, the exchange only changes commodity types and capacities remain honored. Recall that $j$ sends strictly less than $y'_vd_j$ units of flow to $v$, and therefore $(v',j^t)$ was not previously saturated: if we choose a sufficiently small $\epsilon$, its capacity will not be violated. For $(i',k^t)$, its capacity is $y'_i d_k =d_k$ since $i\in I\setminus I_H$; hence its capacity constraint cannot be violated as long as we satisfy the demand constraints.

Note that this change lets $j$ send $\epsilon$ more flow to $S$ while decreasing the total amount of flow $k$ sends to $S$ by $\epsilon$. Since $k$ was sending strictly more than $\alpha d_k$ units of flow to $S$, we can choose $\epsilon>0$ so that both $j$ and $k$ will send strictly more than $\alpha$ fraction of their demands to $S$ after this modification. This either decreases $|J|$ by 1, or increases $\alpha$. Finally, observe that our modification did not introduce any nodes $j^s$ for $j\in\cD\setminus D_H$ back and therefore Claim~\ref{cl:flow1} can still be assumed.



\subsubsection{Proof of Lemma~\ref{lem:cost}: Bounding the Cost of the Semi-integral Solution}\label{sec:rounding}

In this section, we prove Lemma~\ref{lem:cost}. First, we verify that $(\hat{\vx},\hat{\vy})$ is semi-integral.
Let us first verify Conditions~\eqref{cond:semi2} and \eqref{cond:semi3} of semi-integrality. We have $\hat{y}_i=1$ for each $i\in I$ and $\hat{y}_i\leq \frac12$ for each $i\in S=\cF\setminus I$. For any $j\in \cD$, we have
$$\hat{d}_j=\sum_{i\in S} \hat{x}_{ij}=\sum_{i\in S}\frac{d_j\cdot h(i,j)}{h(S,j)}=d_j.$$
Note that $h(S,j)=0$ implies $d_j=0$. For each $j\in \cD$ and $i\in S$, we have $$\hat{x}_{ij}=\frac{d_j\cdot h(i,j)}{h(S,j)}\leq 2h(i,j)=2 \sum_{P\in \cP_{ij}}f(P)\leq 2{y}'_i d_j=\hat{y}_i d_j=\hat{y}_i\hat{d}_j$$
where the first inequality follows from Lemma~\ref{lem:flow} and the last inequality follows from the fact that the capacity of arc $(i',j^t)$ is ${y}'_id_j$ and all paths in $\cP_{ij}$ contain this arc.

Now, let us verify the assignment constraints. For every $j\in \cD$, we have
\begin{align*}
\sum_{i\in \cF} \hat{x}_{ij}=\sum_{i\in I}g^{\star}_{ij}+\sum_{i\in S} \frac{d_j\cdot h(i,j)}{h(S,j)}= \sum_{i\in I}g^{\star}_{ij} +d_j=1
\end{align*}
where the last equality follows from the definition of $d_j$. For every $i\in I$, we have

\begin{align*}
\sum_{j\in \cD} \hat{x}_{ij}=\sum_{j\in \cD} {g}^{\star}_{ij}\leq U_i= U_i\hat{y}_i
\end{align*}
 where the inequality follows from the fact that $\vg^{\star}$ is a fractional matching with bound $U_i$ at facility $i$. For each $i\in S$,
\begin{align*}
\sum_{j\in \cD} \hat{x}_{ij}=\sum_{j\in \cD} d_j\frac{h(i,j)}{h(S,j)}\leq \sum_{j\in \cD} 2 h(i,j)=2\sum_{j\in \cD}\sum_{P\in \cP_{ij}} f(P)\leq 2y^{\star}_i U_i=\hat{y}_i U_i
\end{align*}
where the first inequality again follows from the fact that $h(S,j)\geq \frac{d_{j}}{2}$ and the next inequality follows from the fact all the paths in the sum use the arc $(i,i')$ which has capacity $y_i' U_i$ and $y_i'=y_i^{\star}$ for each $i\in S$.

Now it remains to verify the cost of $(\hat{\vx},\hat{\vy})$. First we have $\hat{y_i}\leq 4 y^{\star}_i$ for each $i\in I$ and $\hat{y}_i=2y^{\star}_i$ for each $i\in S$. Thus we have $\sum_{i\in \cF} o_i \hat{y}_i\leq 4\sum_{i\in \cF} o_iy^{\star}_i$. Now we bound the assignment cost of assignment $\hat{\vx}$. Firstly, the assignment cost to facilities in $I$ can be bounded by the cost of $\vg^{\star}$ which is smaller than cost of $2\vx^{\star}$. To bound the assignment cost to facilities in $S$, let us consider the flow problem and assign cost $c_{ij}$ to each arc $(j^s,i)$ and $(i,j^s)$, and zero to all other arcs. Observe that this together with the triangle inequality implies that every path $P\in \cP_{ij}$ has cost $c(P)\geq c_{ij}$. Thus we have,
\begin{align*}
\sum_{i\in S, j\in \cD} c_{ij} \hat{x}_{ij}\leq \sum_{i\in S, j\in \cD} c_{ij} 2 h(i,j)\leq 2 \sum_{i\in S, j\in \cD} \sum_{P\in \cP_{ij}} c_{ij} f(P) \leq   2 \sum_{i\in S, j\in \cD} \sum_{P\in \cP_{ij}} c(P) f(P)\\
 \leq 2 \sum_{i\in \cF,j\in \cD} c_{ij}\left(  \sum_{P: (j^s,i)\in P } f(P)+\sum_{P: (i,j^s)\in P } f(P)\right) \leq 2\sum_{i\in \cF, j\in \cD} c_{ij} ( x^{\star}_{ij} +2x^{\star}_{ij})= 6\sum_{i\in \cF, j\in \cD} c_{ij} x_{ij}^{\star}
\end{align*}
where the last inequality follows the fact the capacity of arc $(j^s,i)$ is $x^{\star}_{ij}$ and the capacity of arc $(i,j^s)$ is $g^{\star}_{ij}\leq 2x^{\star}_{ij}$.
Thus the total assignment cost is at most $8\sum_{i\in \cF, j\in\cD} c_{ij}x^{\star}_{ij}$. Thus we have $c(\hat{\vx},\hat{\vy})\leq 8 c(\vx^{\star},\vy^{\star})$ and Lemma~\ref{lem:cost} holds.



\begin{figure}[!h]
\centering
\fbox{\parbox{\textwidth}{
\textsl{Input:} Fractional solution $(\vx^{\star},\vy')$.

\textsl{Output:} Either a violating constraint or a semi-integral solution $(\hat{\vx},\hat{\vy})$.

\begin{enumerate}

\item Let $G=(\cD,I,E)$ be the complete bipartite graph whose bipartition is given by $\cD$ and $I$. An arc $(j,i)$ where $j\in \cD$ and $i\in I$ is given a capacity of $2x^{\star}_{ij}$. Every client $j$ has a capacity of one and each facility $i\in I$ is given a capacity of $U_i$.

\item Let $\vec{z}$ denote the maximum fractional capacitated $b$-matching on the capacitated bipartite graph $G$.
\item Let $H$ denote the support of the residual network of the fractional matching $\vec{z}$. In other words, $H=\{(j,i):z_{ij} <2x^{\star}_{ij} \}\cup \{(i,j): z_{ij}>0\}.$
\item Let a client $j\in \cD$ be called \emph{saturated} if $\sum_{i\in I} z_{ij} =1$ and \emph{unsaturated} otherwise. Let
\begin{align*}
I_H &=  \{i \in I: i \mbox{ is reachable in $H$ from some client $k$ that was unsaturated}\}; \\
D_H & = \{j \in \cD: j \mbox{ is reachable in $H$ from some client $k$ that was unsaturated}\}.
\end{align*}

\item  For each $i\in \cF$ and $j\in \cD$, we let
$$g^{\star}_{ij} =
\begin{cases}
z_{ij}&  \mbox{if } i \in I_H\\
z_{ij}& \mbox{if } i\in I\setminus I_H, j\in \cD\setminus D_H\\
0& \mbox{if }i\in I\setminus I_H, j\in  D_H\\
0& \mbox{if }i\in S
\end{cases}
$$

\item Solve $\MFN(\vg^{\star},\vx^{\star},\vy^{\star})$ using Lemma~\ref{lem:linprog}.
\begin{itemize}
\item If Lemma~\ref{lem:linprog} gives a violating constraint, return it.
\item Else let $\vec{f}$ denote the feasible flow satisfying the guarantee in Lemma~\ref{lem:flow}. Then return $(\hat{\vx},\hat{\vy})$ defined as follows.

$$\hat{y}_i =
\begin{cases}
1&  \mbox{if } i \in I\\
2y^{\star}_i& \mbox{if }i\in S
\end{cases}
 ~~~~ \qquad \qquad \qquad~~~~
\hat{x}_{ij} =
\begin{cases}
g^{\star}_{ij}&  \mbox{if } i \in I, j\in \cD\\
d_j\frac{ h(i,j)}{h(S,j)} &\mbox{if } i\in S, j\in \cD
\end{cases}
$$

\end{itemize}
\end{enumerate}
}}
\caption{Overview of algorithm for Theorem~\ref{thm:mainalg2}.}\label{fig:partial}
\end{figure}

\paragraph{Acknowledgments.}
We thank Tim Carnes and David Shmoys for inspiring discussions during the early stages of this work.
We also thank the anonymous reviewers of the conference version of this paper for their helpful comments.

\newpage
\bibliographystyle{plain}
\bibliography{capfl}

\appendix
\section{Formal Proof of Separation Lemma}\label{ap:lp}
We give the formal proof of Lemma~\ref{lem:linprog}.

\begin{replemma}{lem:linprog}
Given $\vec{ g^{\star}}$ in addition to $(\vec x^{\star},\vec y^{\star})$ such that $\MFN(\vec{ g^{\star}},\vec x^{\star},\vec y^{\star})$ is infeasible, we can find in polynomial time a violated inequality, i.e., an inequality that is valid for \MFNLP but violated by $(\vec x^{\star},\vec y^{\star})$. Moreover, the number of bits needed to represent each coefficient of this inequality is polynomial in $|\cF|$, $|\cD|$, and $\log U$, where $U:=\max_{i\in\cF}U_i$.
\end{replemma}

\begin{proof}
  Suppose $\vec{ g}^\star$ is not integral. Since the $b$-matching polytope is integral, we can
  decompose $\vec{g}^\star$ into a convex combination of polynomially many integral $b$-matchings
  $\{\vec{\hat g}_i\}_{1\leq i\leq r}$. Note that such a decomposition can be found in polynomial
  time~\cite{GLS1981}. As was seen in the proof of Lemma~\ref{lem:bing}, at least one of these
  integral $\vec{\hat g}_i$'s renders $\MFN(\vec{\hat g}_i,\vec x^\star,\vec y^\star)$ infeasible. Hence,
  from now on, we will assume $\vec{g}^\star$ is integral, since otherwise we can find in polynomial
  time an integral $\vec{\hat g}_i$ for which $\MFN(\vec{\hat g}_i,\vec x^\star,\vec y^\star)$ is
  infeasible.

  Note that the graph topology of $\MFN(\vec g,\vec x,\vec y)$ does not depend on $(\vec g,\vec
  x,\vec y)$: some arcs may have zero capacities depending on $(\vec g,\vec x,\vec y)$, but the
  underlying digraph of $\MFN(\vec g,\vec x,\vec y)$ is defined independently from $(\vec g,\vec
  x,\vec y)$. It is only the capacities of these arcs that are defined as (linear) functions of
  $(\vec g,\vec x,\vec y)$. Let $\cA$ denote the set of arcs in $\MFN(\cdot,\cdot,\cdot)$ and
  $\cP_j$ the family of all $j^s - j^t$ paths. Let $c_a(\vec g,\vec{x}, \vec{y})$ be the capacity of
  arc $a\in \cA$ in $\MFN(\vec g,\vec x,\vec y)$ and $d_j(\vec{g})$ be the demand of client
  $j$. Note that $d_j(\vec{g})$ does not depend on $\vec x$ or $\vec y$.

  As Onaga~\cite{Onaga70} and Iri~\cite{Iri71} showed, it follows from Farkas' lemma that
  $\MFN(\vec{g}^\star,\vec x,\vec y)$ is feasible if and only if for all $\vec z\in\mathbb{R}_+^\cD$
  and $\vec \ell\in\mathbb{R}_+^\cA$ satisfying\begin{align}
    &z_j \leq \sum_{a\in P} \ell_a,  \qquad \forall j\in \cD, P\in \cP_j;\textrm{ and}\label{e:seplp:1}\\
    & \vec{0} \leq \vec{z}, \vec{\ell} \leq \vec{1},\label{e:seplp:2}
\end{align}the following holds:\begin{equation}\label{e:mfdual}
\sum_{j\in \cD}d_j(\vec{g}^\star) z_j  \leq \sum_{a\in \cA} c_a(\vec{g}^\star, \vec{x}, \vec{y}) \ell_a
.\end{equation}

Thus, as $\MFN(\vec{g}^\star,\vec x^\star,\vec y^\star)$ is infeasible, there exists $(\vec z^\star,\vec l^\star)$ such that\begin{align*}
             &z^\star_j \leq \sum_{a\in P} \ell^\star_a,  & \forall j\in \cD, P\in \cP_j;\\
             & \vec{0} \leq \vec{z}^\star, \vec{\ell}^\star \leq \vec{1};&\\
&\sum_{j\in \cD}d_j(\vec{g}^\star) z^\star_j  > \sum_{a\in \cA} c_a(\vec{g}^\star, \vec{x}^\star, \vec{y}^\star) \ell^\star_a.&
\end{align*}Our separation oracle finds such $(\vec z^\star,\vec l^\star)$ by solving the following LP:
\begin{align*}
\textrm{minimize }\quad & \sum_{a\in \cA} c_a(\vec{g}^\star, \vec{x}^\star, \vec{y}^\star) \ell_a - \sum_{j\in \cD}d_j(\vec{g}^\star) z_j,\\
\textrm{subject to }\quad &\eqref{e:seplp:1}-\eqref{e:seplp:2}.
\end{align*}Note that this LP can be solved in polynomial time by using Djikstra's algorithm as the separation oracle. We will in particular find an extreme point solution to this LP.

Now we have\[
\sum_{j\in \cD}d_j(\vec{g}^\star) z^\star_j  > \sum_{a\in \cA} c_a(\vec{g}^\star, \vec{x}^\star, \vec{y}^\star) \ell^\star_a
;\]but on the other hand, any feasible solution $(\vec x,\vec y)$ to \MFNLP has to satisfy\begin{equation}\label{e:ve}
\sum_{j\in \cD}d_j(\vec{g}^\star) z^\star_j  \leq \sum_{a\in \cA} c_a(\vec{g}^\star, \vec{x}, \vec{y}) \ell^\star_a,
\end{equation}since this is a necessary condition for $\MFN(\vec{g}^\star,\vec x,\vec y)$ to be feasible. Thus we output \eqref{e:ve} as a violated inequality. Recall that $c_a(\vec{g}^\star, \vec{x}, \vec{y})$ is a linear function in $(\vec x,\vec y)$; hence \eqref{e:ve} is a linear inequality in $(\vec x,\vec y)$.

Now it remains to verify that each coefficient of \eqref{e:ve} can be represented in
$\poly(|\cF|,|\cD|,\log U)$ bits. $d_j(\vec g)$ is a constant that is either 0 or 1;
$c_a(\vec{g}^\star, \vec{x}, \vec{y})$ is a linear function in $(\vec x,\vec y)$ where every
coefficient is an integer between 0 and $U$. Finally, note that every coefficient of
\eqref{e:seplp:1} and \eqref{e:seplp:2} are $O(1)$; hence, as we have chosen $(\vec z^\star,\vec \ell^\star)$ as
an extreme point solution, $(\vec z^\star,\vec \ell^\star)$ can be represented in $\poly(|\cF|,|\cD|)$
bits. Thus, \eqref{e:ve} can be written as a linear inequality in $(\vec x,\vec y)$ where every
coefficient is represented using $\poly(|\cF|,|\cD|,\log U)$ bits.
\end{proof}

\section{Rounding a Semi-Integral Solution to an Integral Solution}\label{sec:soft}

We present in this appendix how to round a semi-integral solution to an integral solution using known bi-criteria algorithms for the soft-capacitated problem with general demands.

\begin{replemma}{lemma:main3}
 Given a semi-integral solution $(\hat{\vx},\hat{\vy})$, we can in polynomial time find  an integral solution $(\bar{\vx},\bar{\vy})$ whose cost is at most $36 c(\hat{\vx},\hat{\vy})$.
\end{replemma}

\begin{proof}
 Let $(\hat{\vx},\hat{\vy})$ be a semi-integral solution. We consider the residual instance on facilities in $S$, where each client $j\in \cD$ is assigned to these facilities by the fraction of $\hat d_j=\sum_{i\in S} \hat{x}_{ij}$. Using this residual solution, we now formulate an instance 
 with general demands as follows.

 In the capacitated facility location problem with general demands, in addition to the standard input of the \cfl, we are also given $\bar d_j\geq 0$ for each client $j$ and goal is to open a subset of facilities and assign clients to facilities (possibly by splitting demand among multiple facilities) such that total demand assigned to an open facility is no more than its capacity. The objective function is the sum of cost of opened facilities and cost of the assignment. The cost of assigning $d$ amount demand of client $j$ to facility $i$ is $d\cdot c_{ij}$.  Figure~\ref{fig:lpdemand} shows the standard LP relaxation for the facility location problem with general demands. We will call this program $\LPdemand$. We denote the capacity of facility $i$ by $U_i'$ to differentiate it from the capacities in our problem.
\begin{figure}[h!]
\begin{equation*}
\boxed{%
	\begin{array}{lrclr}
          \textrm{minimize}   &  \multicolumn{4}{l}{\displaystyle \sum_{i\in \cF} o_i y_i + \sum_{i \in \cF, j \in \cD}c_{ij}x_{ij}}     \\[8mm]
	\textrm{subject to} &\displaystyle \sum_{i \in \cF} x_{ij} & \displaystyle=
        &\displaystyle \bar d_j, &\qquad \displaystyle \forall j \in
        \cD\\[4mm]
 &\displaystyle \bar d_j y_i & \displaystyle \geq & \displaystyle x_{ij}, & \displaystyle\qquad \forall i\in \cF, j \in \cD\\[4mm]
 & \displaystyle U_iy_i & \displaystyle \geq & \displaystyle \sum_{j\in \cD} x_{ij}, & \displaystyle\qquad \forall i\in \cF\\[4mm]
  & \displaystyle y_i & \displaystyle \leq & 1, & \displaystyle\qquad \forall i\in \cF\\[4mm]
& \displaystyle \vec{x},\vec{y} &\displaystyle    \geq & \displaystyle 0.
	\end{array}
}
\end{equation*}
\caption{$\LPdemand$.}
\label{fig:lpdemand}
\end{figure}

The following theorem follows from Abrams et al~\cite{AMMP2002}.

\begin{theorem}[\cite{AMMP2002}]\label{thm:abrams}
Given a feasible solution $(\vx,\vy)$ to $\LPdemand$, there exists a solution $(\brx,\bry)$ such that $\bry$ is integral and satisfies $\LPdemand$ if $U_i'$ is replaced by $2U_i'$. Moreover, $c(\brx,\bry)\leq 18c(\vx,\vy)$.
\end{theorem}

We now complete the proof of Lemma~\ref{lemma:main3} using Theorem~\ref{thm:abrams}.
Consider the facility location problem with fractional demands created by facilities in $S$ and demand $\bar d_j=\hat d_j =\sum_{i\in S} \hat{x}_{ij}\leq 1$ for each client $j\in \cD$. Let $U_i'=\frac{U_i}{2}$. We decompose the semi-integral solution $(\hat\vx,\hat\vy)$ into $(\hat\vx^S,\hat\vy^S)\in[0,1]^{S\times\cD}\times[0,\textstyle\frac{1}{2}]^S$ and $(\hat\vx^I,\hat\vy^I)\in[0,1]^{I\times\cD}\times\{1\}^I$ by naturally restricting $(\hat\vx,\hat\vy)$ to $I$ and $S$, respectively. Then $(\hat{\vx}^S, 2\hat{\vy}^S)$ is a feasible solution to $\LPdemand$. Observe that here we use the fact that $\hat{y}^S_i\leq \frac{1}{2}$ for each $i\in S$ and therefore we have $2\hat{y}^S_i\leq 1$ and the upper bound constraints are satisfied. Feasibility of the rest of the constraints follows from semi-integrality of $(\hat{\vx},\hat{\vy})$. We run the algorithm from Theorem~\ref{thm:abrams} on $(\hat{\vx}^S, 2\hat{\vy}^S)$, and it returns a solution $(\brx^S,\bry^S)$ such that $\bry^S\in\{0,1\}^S$ is integral and capacity constraints are satisfied with respect to $2U_i'=U_i$. Moreover, $c(\brx^S,\bry^S)\leq 18c(\hat{\vx}^S, 2\hat{\vy}^S)\leq 36c(\hat{\vx}^S, \hat{\vy}^S)$.

We concatenate $(\hat\vx^I,\hat\vy^I)\in[0,1]^{I\times\cD}\times\{1\}^I$ and $(\brx^S,\bry^S)\in[0,1]^{S\times\cD}\times\{0,1\}^S$ to obtain $(\brx',\bry)\in[0,1]^{\cF\times\cD}\times\{0,1\}^\cF$. Observe that now $(\brx',\bry)$ satisfies the following constraints.
\begin{enumerate}
\item\label{c:final:q:1} $\bry$ is integral.
\item\label{c:final:q:2} For each $j\in \cD$, we have $\sum_{i\in \cF}\bar{x}'_{ij}= \sum_{i\in I} \hat x^I_{ij}+ \sum_{i\in S} \bar y^S_{ij}=1-\hat d_j+\bar d_j=1$.
\item\label{c:final:q:3} For each $i\in \cF$, we have $\sum_{j\in \cD}\bar{x}'_{ij}\leq U_i$ if $\bar{y}_i=1$ and $\sum_{j\in \cD}\bar{x}'_{ij}=0$ otherwise.
\end{enumerate}

Now we solve a minimum cost $b$-matching problem to find an integral assignment of clients to facilities which are opened by $\bry$. We let $F$ denote the set of facilities with $\bar{y}_i=1$. We form a bipartite graph with clients in $\cD$ on one side and facilities in $F$ on the other side. The cost of edge between $i\in F$ and $j\in \cD$ is $c_{ij}$ and solve the following linear program as given in Figure~\ref{fig:bmatching}. It is straightforward to see that $\brx'$ is a feasible solution to this linear program. The integrality of the $b$-matching problem implies that there exists an integral solution $\brx$ such that $c(\brx)\leq c(\brx')$.
\begin{figure}[h!]
\begin{equation*}
\boxed{%
	\begin{array}{lrclr}
          \textrm{minimize}   &  \multicolumn{4}{l}{\displaystyle \sum_{i \in F, j \in \cD}c_{ij}x_{ij}}     \\[8mm]
	\textrm{subject to} &\displaystyle \sum_{i \in F} x_{ij} & \displaystyle=
        &\displaystyle 1, &\qquad \displaystyle \forall j \in
        \cD\\[4mm]
 & \displaystyle \sum_{j\in \cD} x_{ij}  & \displaystyle \leq &  \displaystyle U_i,& \displaystyle\qquad \forall i\in F\\[4mm]
& \displaystyle \vec{x} &\displaystyle    \geq & \displaystyle 0.
	\end{array}
}
\end{equation*}
\caption{$b$-matching Linear Program.}
\label{fig:bmatching}
\end{figure}

Together with $\bry$, we obtain that $(\brx,\bry)$ is a feasible solution to the capacitated facility location problem and\begin{eqnarray*}
c(\brx,\bry)&\leq& c(\brx',\bry)=c(\hat\vx^I,\hat\vy^I)+c(\brx^S,\bry^S)\leq c(\hat\vx^I,\hat\vy^I)+36c(\hat{\vx}^S, \hat{\vy}^S)
\leq 36(c(\hat\vx^I,\hat\vy^I)+c(\hat{\vx}^S, \hat{\vy}^S))\\
&\leq& 36c(\hat{\vx},\hat{\vy})
\end{eqnarray*}
as claimed.
\end{proof}


\end{document}